\documentclass[abstract]{scrartcl}

\usepackage[T1]{fontenc}
\usepackage[utf8]{inputenc}
\usepackage{amsmath}
\usepackage{amsfonts}
\usepackage{amssymb}
\usepackage{amsthm}
\usepackage{graphicx}
\usepackage{appendix}
\usepackage[english]{babel}
\usepackage{mathtools}
\usepackage{enumerate}

\numberwithin{equation}{section}

\usepackage[breaklinks=true]{hyperref}
\hypersetup{ 
	colorlinks,
	linkcolor={blue!60!black},
	citecolor={blue!60!black},
	urlcolor={blue!60!black},
	linktoc=page}

\usepackage[numbers]{natbib}
\usepackage{tikz-cd}

\usepackage{arydshln} 

\usepackage{authblk} 

\usepackage{todonotes}

\newcommand{\defn}{\emph}

\renewcommand{\d}{\mathrm{d}}

\newcommand{\liec}{\mathrm{C}}

\newcommand{\cF}{\mathcal{F}}
\newcommand{\cJ}{\mathcal{J}}
\newcommand{\cL}{\mathcal{L}}

\newcommand{\var}[3]{\frac{\delta_{#1} {#2}}{\delta {#3}}}
\newcommand{\der}[2]{\frac{\partial {#1}}{\partial {#2}}}
\newcommand{\lda}{\lambda}

\newcommand{\N}{\mathbb{N}}
\newcommand{\R}{\mathbb{R}}

\newcommand{\g}{\mathfrak{g}}
\newcommand{\h}{\mathfrak{h}}

\newcommand{\1}{\mathbf{1}}

\DeclareMathOperator{\D}{D}
\DeclareMathOperator{\pr}{pr}
\DeclareMathOperator{\tr}{tr}

\newcommand{\ip}{\lrcorner\,}

\newtheorem{thm}{Theorem}[section]
\newtheorem{prop}[thm]{Proposition}
\newtheorem{lemma}[thm]{Lemma}

\newtheorem{remark}[thm]{Remark}

\theoremstyle{definition}
\newtheorem{definition}[thm]{Definition}
\newtheorem{example}[thm]{Example}

\title{Lagrangian multiforms on Lie groups and non-commuting flows}

\author{Vincent Caudrelier}
\author{Frank Nijhoff}
\author{Duncan Sleigh,}
\affil{\large School of Mathematics, University of Leeds, Leeds, UK.
\\ \texttt{v.caudrelier@leeds.ac.uk, f.w.nijhoff@leeds.ac.uk, mm16dgs@leeds.ac.uk}}
\author{Mats Vermeeren\thanks{Corresponding author}}
\affil{\large Department of Mathematical Sciences, Loughborough University, \\ Loughborough, UK.
\\ \texttt{m.vermeeren@lboro.ac.uk} }

\date{}

\begin{document}

\maketitle

\begin{abstract}
\noindent
We describe a variational framework for non-commuting flows, extending the theories of Lagrangian multiforms and pluri-Lagrangian systems, which have gained prominence in recent years as a variational description of integrable systems in the sense of multidimensional consistency. In the context of non-commuting flows, the manifold of independent variables, often called multi-time, is a Lie group whose bracket structure corresponds to the commutation relations between the vector fields generating the flows. Natural examples are provided by superintegrable systems for the case of Lagrangian 1-form structures, and integrable hierarchies on loop groups in the case of Lagrangian 2-forms. As particular examples we discuss the Kepler problem, the rational Calogero-Moser system, and a generalisation of the Ablowitz-Kaup-Newell-Segur system with non-commuting flows. We view this endeavour as a first step towards a purely variational approach to Lie group actions on manifolds.   



\end{abstract}

\section{Introduction} 

The variational theory behind integrable systems has recently undergone a major development through the introduction of Lagrangian multiform theory \cite{lobb2009lagrangian}, which establishes 
a variational framework for \emph{multidimensional consistency}, a notion which generalises the idea of commuting flows of differential equations, including the analogous property for lattice equations. Multidimensional consistency is the key integrability phenomenon of the coexistence and compatibility of a multitude of dynamical equations on the same dependent variable in terms of several (often an arbitrary number) independent 
variables. Thus, multidimensional consistency is a manifestation of well-known integrability aspects, such as the existence of hierarchies of nonlinear evolution equations, their infinite sequences of conservation laws, associated linear problems (Lax pairs), dressing transforms and Darboux schemes. 

A conventional Lagrangian formalism only provides a single (Euler-Lagrange) equation per component of the relevant field variables. From the perspective of multi-dimensional consistency, it is more natural to consider a space of independent variables of arbitrary 
dimension, called \emph{multi-time}. In the \emph{multiform} formalism (and in the closely related \emph{pluri-Lagrangian} formalism),  Lagrangians are components of a 
differential (or difference) $d$-form in multi-time. The associated Euler-Lagrange equations provide a compatible system of simultaneous equations on each component of the fields. 

Since the initial proposal in \cite{lobb2009lagrangian}, this new variational approach has been shown to be a quite universal structure behind integrable systems, both in the realm of discrete and continuous equations of this class, see e.g.\@ \cite{lobb2010kp,lobb2010gelfand-dikii,xenitidis2011Lagrangian-structure,yoo2011discrete,boll2014integrability,sleigh2019variational,sleigh2021lagrangian}. In the continuous case, the full set of variational equations, i.e. the
extended set of Euler-Lagrange equations, were worked out in \cite{suris2016lagrangian}. 
Many salient features have been elaborated in recent years, e.g.\@ the connection with variational symmetries and Noether's theorem \cite{petrera2017variational,vermeeren2021hamiltonian,sleigh2020variational}, as well as with Lax pairs \cite{sleigh2019variational}, Hamiltonian structures \cite{suris2013variational, caudrelier2020hamiltonian, vermeeren2021hamiltonian} and the classical $r$-matrix \cite{caudrelier2021multiform}. 

So far the Lagrangian multiform approach has proven successful in providing a framework for multi-time integrable systems in terms of commuting flows, which is the traditional setting of multidimensionally consistent integrable systems. However,  the notion of integrability is not necessarily restricted to the case of commuting flows. In fact, the pioneering paper \cite{mishchenko1978generalized} strayed beyond the confines of flows generated by Hamiltonians in involution with regard to the Poisson structure, and considered non-commuting flows as well, establishing a Liouville type theorem for the case that there is a nontrivial Lie algebra structure for the corresponding vector fields. 

In \cite{nijhoff1988integrable} non-commuting flows on loop algebras were considered, generating systems integrable through linear integral equations associated with Lax representations of $N\times N$ matrix hierarchies. Corresponding Lagrangian structures were established in \cite{nijhoff1986integrable}. In these extended integrable systems, compatibility relations generate vector fields 
that are `alien' to the original flows in the equations. Hence any variational description would necessarily entail a multiform structure. In the present paper we aim at providing a solid 
underpinning of these ideas, which amount to a description of dynamical systems where the set of \emph{independent variables} is a Lie group\footnote{In this respect what is aimed at here is essentially different from recent work on non-commutative integrable systems, where the \emph{dependent variables} are chosen in some associative algebra.}. 
In contrast, so far Lie groups have mostly appeared in dynamical systems either as objects describing the symmetries or as the phase space. What is done in the present paper is to consider a multi-time (space of independent variables) that has the structure of a Lie group.  
More ambitiously, the present set-up could be viewed as a first step in general towards a purely variational approach to Lie group actions on manifolds. 

The outline of the paper is as follows. In section 2 we provide some background on the Lagrangian structure of Liouville integrable systems and sketch the multiform approach. In section 3 we set up the framework for the Lagrangian multiform structure for systems where a Lie group is considered as a (non-commuting) multi-time manifold. In section 4 we provide the corresponding variational structure for those Lagrange 1-forms and provide some compelling examples: the Kepler problem and the rational Calogero-Moser system. These systems have the added feature that they are superintegrable, so a full description of the group actions of their symmetries necessarily requires non-commuting flows.
In section 5 we consider the case of Lagrange 2-forms on loop groups, building on the structures of  \cite{nijhoff1986integrable,nijhoff1988integrable} as well as \cite{caudrelier2021multiform}, and present, as a new example, a non-commutative multi-time version of the AKNS hierarchy. We end with some conclusions in section 6. 

\section{Background: Lagrangian structure of Liouville integrable systems}
\label{sec-background}

A Hamiltonian system on a $2n$-dimensional symplectic manifold $S$ is Liouville integrable if the Hamiltonian function $H_1:S \rightarrow \R$ is one of $n$ functionally independent functions $H_i:S \rightarrow \R$ ($i=1,\ldots,n$) such that $\{H_i,H_j\} = 0$ for all $i,j$, where $\{\cdot,\cdot\}$ is the Poisson bracket induced by the symplectic structure $\omega$ on $S$. The flows of the additional Hamiltonian functions $H_2,\ldots,H_n$ are symmetries of the Hamiltonian system $(S,H_1)$. The Liouville-Arnold theorem shows that this setup has a rich geometric structure (see e.g.\@ \cite{babelon2003introduction}). Below we will assume that $S$ is a cotangent bundle, $S = T^* Q$.

It is important to note that the definition of Liouville integrability is symmetric under relabelling of the Hamiltonian functions $H_i$. It makes no difference which Hamiltonian function we consider to be physical (all others being its symmetries). For each $i$ we could define a flow 
\[ \Phi_i : \R \times S \rightarrow S: (t,z) \mapsto \Phi_i^t(z) \,,\]
where $\der{\Phi_i^t(z)}{t} \ip \omega = \d H_i$. But why should we consider these flows as separate objects? We might as well introduce a ``flow'' on the \defn{multi-time} $\R^n$,
\begin{equation}\label{multi-flow}
\Phi : \R^n \times S \rightarrow S: (t_1,\ldots,t_n,z) \mapsto \Phi^{t_1} \circ \ldots \circ \Phi^{t_n} (z) \,,
\end{equation}
which captures the combined dynamics of the system and its symmetries (in the sense of the action of the symmetry group on phase space).

The advantage of combining the physical time and the ``times'' of the symmetry flows into multi-time manifests itself more clearly in the Lagrangian picture. Assuming the Hamiltonian functions $H_1,\ldots,H_n$ are non-degenerate%
\footnote{In fact they need not all be non-degenerate, but rather non-degenerate as a family. See \cite{suris2013variational,vermeeren2021hamiltonian}.}%
, we could introduce Lagrangians $L_1, \ldots, L_n$. Then the flows $\Phi_i: \R \times T^* Q \rightarrow T^* Q$, projected down to $Q$, produce critical curves of the action integrals
\[ \int_a^b L_i(q,\dot{q}) \,\d t \,.\]
In the multi-time formalism we can combine the Lagrangian functions into a 1-form
\[ \cL[q] = L_1[q] \,\d t_1 + \ldots + L_n[q] \,\d t_n \,,\]
where the square brackets denote dependence on a function $q: \R^n \rightarrow Q$ and its derivatives. We call such a function, from multi-time to the configuration manifold, a \defn{field}. The main feature of a 1-form is that it can be integrated along curves. So for every curve $\gamma: [0,1] \rightarrow \R^n$ we can define an action functional
\[ A_\gamma: q \mapsto \int_\gamma \cL[q] \,. \]
We can now impose the following variational principle, which provides the setting to recognise integrability from the Lagrangian perspective \cite{lobb2009lagrangian, yoo2011discrete, suris2013variational, suris2016lagrangian}.
\begin{definition}
\label{def:pluri-lag}
We say that a field $q: \R^n \rightarrow Q$ is \defn{critical} for $\cL$ if the corresponding action $A_\gamma$ is critical for every curve $\gamma$. That is, for every $\gamma: [0,1] \rightarrow \R^n$ and every
 smooth family of curves $q_\varepsilon$ such that $q = q_0$ and $q_\varepsilon(\gamma(0)) = q(\gamma(0))$ and $q_\varepsilon(\gamma(1)) = q(\gamma(1))$ there holds
\[ \frac{\d}{\d \varepsilon} \Big|_{\varepsilon = 0} \int_\gamma \cL[q_\varepsilon] = 0 \,. \]
If the Lagrangian depends on second or higher derivatives, we also require the derivatives of $q$ and $q_\varepsilon$ to be equal at the endpoints of $\gamma$
\end{definition}
A system that is described by a Lagrangian 1-form via this variational principle is known in the literature as a ``pluri-Lagrangian system'' \cite{boll2014integrability, bobenko2015discrete, suris2016lagrangian}. Additionally, one often requires that $q$ is critical with respect to variations of the curve $\gamma$ too, which is equivalent to requiring that $\cL[q]$ is closed. This perspective is known as ``Lagrangian multiform'' theory \cite{lobb2009lagrangian, xenitidis2011Lagrangian-structure, hietarinta2016discrete}. The closure property $\d \cL =0$ implies that the corresponding Hamiltonian functions are in involution \cite{suris2013variational,vermeeren2021hamiltonian}. In addition to being a formalism to derive equations from a given Lagrangian, Lagrangian multiform theory can be seen as a guiding principle to determine integrable Lagrangians. 

\paragraph{Multi-time Euler-Lagrange equations} The differential equations which characterise criticality in the sense of Definition \ref{def:pluri-lag} are called \defn{multi-time Euler-Lagrange equations} (or \defn{multiform Euler-Lagrange equations}). They were first derived in \cite{suris2016lagrangian} and using a different approach in \cite{sleigh2020variational,sleigh2021lagrangian}. Below we give a heuristic explanation of this system of equations.

If we choose the curve $\gamma$ to be a straight line in the $t_i$-direction, we recover a familiar action integral with the corresponding component of $\cL$ as the Lagrangian: $\int L_i \,\d t_i$. This leads to the Euler-Lagrange equation
\[  \der{L_i}{q} - \D_i \der{L_i}{q_i} + \ldots = 0 \,, \]
where $D_i$ denotes the total derivative with respect to $t_i$, $q_i = \frac{\d q}{\d t_i}$, and the dots represent terms of the Euler-Lagrange equation which are relevant if the Lagrangian depends on second or higher derivatives.
We call this expression a \defn{variational derivative} and denote it by
\[ \var{i}{L_i}{q} = \der{L_i}{q} - \D_i \der{L_i}{q_i} + \ldots \,, \]
where the first index $i$ indicates that additional derivatives, which originate from integration by parts in the standard derivation of the Euler-Lagrange equations, are only with respect to $t_i$. Because $\gamma$ was taken in the $t_i$-direction, these are the only integrations by parts that could be carried out. In particular, derivatives of $q$ with respect to other time variables cannot be integrated away. Therefore we should consider them as additional variables and include the corresponding Euler-Lagrange equations
\begin{align*}
\var{i}{L_i}{q_j} &= \der{L_i}{q_j} - \D_i \der{L_i}{q_{ji}} + \ldots = 0 \,, \\
\var{i}{L_i}{q_{jk}} &= \der{L_i}{q_{jk}} - \D_i \der{L_i}{q_{jki}} + \ldots = 0 \,, \\
&\vdotswithin{=}
\end{align*}
where $j,k,\ldots \neq i$ and subscripts of $q$ denote partial derivatives.

So far we have only considered curves $\gamma$ which are in coordinate directions. Additional multi-time Euler-Lagrange equations are found when we consider curves in other directions (or curves that are not straight). They are of the form
\[
\var{i}{L_i}{q_i} = \var{j}{L_j}{q_j} \,, \qquad
\var{i}{L_i}{q_{ki}} = \var{j}{L_j}{q_{kj}} \,, \qquad
\ldots
\]

In summary, the variational principle of Definition \ref{def:pluri-lag} is equivalent to the following system of multi-time Euler-Lagrange equations:
\begin{align*}
    &\var{i}{L_i}{q_I} = 0 \,,  \qquad\qquad I \not\ni t_i \,, \\
    &\var{i}{L_i}{q_{Ii}} = \var{j}{L_j}{q_{Ij}} \,,
\end{align*}
where $I$ is a multi-index listing the differentiations applied to $q$, $I \not\ni t_i$ means that none of them are with respect to $t_i$, and the $i$ in $q_{Ii}$ denotes an additional differentiation with respect to $t_i$.

An even more compact expression for the multi-time Euler-Lagrange equations can be given as
\[  \var{i}{L_i}{q_{I \setminus j}} = \var{j}{L_j}{q_{I \setminus i}} \,,  \]
where $I \setminus i$ denotes one fewer differentiation with respect to $t_i$. If $I$ does not list any differentiations with respect to $t_i$, then any term containing $I \setminus i$ is taken to be zero.

\paragraph{Exterior derivative}

The variational principle of Definition \ref{def:pluri-lag} gives a single Lagrangian description of a number of commuting flows. To capture integrability in the sense of Liouville, we need more than commutativity of the flows. (Commutativity corresponds to constant, not necessarily vanishing, Poisson brackets.) The key integrability feature in the multi-form approach is the \defn{closure relation}: the exterior derivative $\d \cL$ should vanish when evaluated on solutions of the multi-time Euler-Lagrange equations.

Furthermore, taking variations of (the coefficients of) $\d \cL$ is equivalent to the variational principle. Hence $\d \cL$ is zero on solutions if and only if it attains a double zero on solutions. In many examples one can write the coefficients of $\d \cL$ explicitly as a product of two expressions which vanish on the multi-time Euler-Lagrange equations.

The closedness of the form $\cL$ is immediately related to the vanishing of Poisson brackets between the corresponding Hamiltonian functions \cite{suris2013variational,vermeeren2021hamiltonian}, and to the fact that the commuting flows are variational symmetries of each other's Lagrangian functions \cite{sleigh2020variational,petrera2021variational}.

\paragraph{Higher forms}
 
So far in this introduction we have only mentioned the multi-form principle for 1-forms, which applies to systems of ODEs. In the case of hierarchies  of PDEs, a completely analogous principle applies for a higher form. For example, in integrable hierarchies such as KdV \cite{suris2016lagrangian} and AKNS \cite{sleigh2019variational,sleigh2020variational}, the individual equations are 2-dimensional, so the classical variational principle involves integration over a plane. In the multi-time setting, all the equations of such a hierarchy share the same space variable, but they each have their own time variable. Multi-time is spanned by the full set of space and time directions. The pluri-Lagrangian principle now requires that the integral of a 2-form is critical regardless of which 2-dimensional surface of integration is chosen. As before, the Lagrangian multiform principle augments this by the fact that the action should also be critical with respect to variations of the surface of integration. For higher-dimensional PDEs one can consider higher forms. For example, there is a Lagrangian 3-form description of the KP hierarchy \cite{sleigh2021lagrangian}.

\section{A Lie group as multi-time}

Many systems have symmetries that do not all commute with each other. Of particular interest are those Hamiltonian systems where there exist functions $H_1,\ldots,H_{n+\ell}:T^*Q \rightarrow \R$ such that
\[ \{ H_i, H_k \} = 0 \qquad \text{for all }i \in \{1,\ldots,n-\ell\} \text{ and all }k \in \{1,\ldots,n+\ell\} \,, \]
and the remaining Poisson brackets may be nonzero. Systems like this are called \emph{non-commutative integrable} or \emph{degenerate integrable} and a simple adaptation of the Liouville-Arnold theorem applies to them \cite{mishchenko1978generalized, bolsinov2003noncommutative}. For $\ell=0$ we recover Liouville integrability, and it can be shown that these conditions for $\ell>0$ also imply Liouville integrability \cite{bolsinov2003noncommutative}, hence the term \emph{superintegrability} is also used for such systems \cite{fasso2005superintegrable}. Of course, one may also be interested in non-integrable systems that still possess some smaller amount of symmetries with nontrivial commutation relations.

If some of the $H_k$ have non-constant Poisson brackets, their flows will not commute. Hence we cannot consider the flows of all $H_1,\ldots,H_{n+\ell}$ together as functions of some multi-time $\R^{n+\ell}$. Indeed equation \eqref{multi-flow} breaks down because it now depends on the order in which we list the flows $\Phi^{t_i}$, which defeats the point of putting all flows on the same footing. But all is not lost. The infinitesimal generators of the flows of the $H_i$ form a Lie algebra. We can use a copy of the (universal covering) Lie group $G$ of this Lie algebra as multi-time. The Hamiltonian ``flow'' on multi-time now depends on a Lie group element $g$ instead of a number of time coordinates, 
\begin{equation}
 \Phi: G \times T^*Q \rightarrow T^*Q: (g,z_0) \mapsto \Phi^g(z_0) = g \cdot z_0 \,,
 \end{equation}
where $\cdot$ denotes the left group action of $G$ on $T^*Q$ by symplectomorphisms.
Hence $\Phi$ assigns to an initial condition $z_0 \in T^*Q$ its flow under an element $g \in G$ of multi-time. In case all symmetries commute, $g$ would be the vector of times $(t_1,\ldots,t_n)$.

Unlike dynamical systems where the Lie group is the phase space (such as e.g.\@ in rigid body dynamics, \cite{abraham2008foundations,marsden2013introduction}), in this setting it is the space of independent variables that possesses the structure of a Lie group $G$. Here we take phase space to be a cotangent bundle $T^*Q$ for which no additional structure is assumed.

We would like to think of the Lie group $G$ not just as the multi-time, but also as a symmetry group acting on fields, with $g \in G$ acting on a field $z: G \rightarrow T^*Q$ to produce a new field $(g \bullet z): G \rightarrow T^*Q$ defined by 
\[ (g \bullet z)(h) = z(g h) \,. \]
Note that we defined both actions, $\cdot$ and $\bullet$, as left actions. Alternatively, we could have adopted a convention where both are right actions.

We are interested in those fields $z: G \rightarrow T^*Q$ for which the two actions of $G$ agree: the flow over ``time'' $g \in G$ maps the field to its transformation by $g$. We call such fields ``symmetry group solutions'' and define them as follows.
\begin{definition}
Let $G$ act by symplectic transformations on $T^*Q$ and denote this action by $\cdot$.
We say that a field $z: G \rightarrow T^*Q$ is a \defn{symmetry group solution} of this group action if for all $h \in G$ there holds $\Phi^g ( z(h) ) = (g \bullet z)(h)$ or, equivalently,
\begin{equation}\label{sym-gp-sol}
g \cdot z(h)  = z(g h),
\end{equation}
\end{definition}

\begin{example}
Let $G = (\R,+)$ act on $\R^2$, with coordinates $(q,p)$ by horizontal translation:
\[ \Phi^g(q,p) = g \cdot (q,p) = (q+g,p) \,. \]
Note that $\Phi$ is the flow of the Hamiltonian $H(q,p) = p$, assuming the standard symplectic structure.
For any constants $q_0$, $p_0$ we have a symmetry group solution $z_s:G \rightarrow \R^2$ defined as
\[ z_s(g) = (q_0 + g, p_0) \,. \]
Indeed we have
\[ (g \bullet z_s)(h) = (q_0 + g + h, p_0) = g \cdot z_s(h) = \Phi^g ( z_s(h) ) \,. \]
Any function $z:G \rightarrow \R^2$ that is not of this form will not be a symmetry group solution. As a specific counterexample, consider 
\[ z_c(g) = (q_0, p_0 + g) \,. \]
We have 
$ (g \bullet z_c)(h) = (q_0, p_0 + g + h) $
but
$ \Phi^g ( z_c(h) ) = g \cdot z_c(h) = (q_0 + g, p_0 + h)$.
\end{example}

\begin{example}
\label{ex-SE2}
As an example involving a nonabelian group, consider the Lie group $SE(2)$, parameterised by $(x,y,\theta)$, with multiplication
\begin{equation}
\label{SE2-product}
(x',y',\theta') (x,y,\theta) = (x' + x \cos \theta' - y \sin \theta', y' + x \sin \theta' + y \cos \theta' , \theta + \theta' ) \,.
\end{equation}
It acts by Euclidean transformations of the $(q,p)$-plane:
\begin{equation}
\label{SE2-action}
\Phi^{(x,y,\theta)}\begin{pmatrix} q \\ p \end{pmatrix}  =  (x,y,\theta) \cdot \begin{pmatrix} q \\ p \end{pmatrix}
= \begin{pmatrix} \cos \theta & - \sin \theta \\ \sin \theta &\cos \theta \end{pmatrix}\begin{pmatrix} q \\ p \end{pmatrix}
+ \begin{pmatrix} x \\ y \end{pmatrix} \,.
\end{equation}
Note that the one-parameter flow maps $\Phi^{(x,0,0)}$, $\Phi^{(0,y,0)}$ and $\Phi^{(0,0,\theta)}$ correspond to the Hamiltonian systems given by $H(q,p) = p$, $H(q,p) = -q$ and $H(q,p) = -\frac12(p^2 + q^2)$, respectively.

Any function $z: SE(2) \to \R^2$ of the form 
\[z(x,y,\theta) = (x,y,\theta) \cdot \begin{pmatrix} q_0 \\ p_0 \end{pmatrix},\]
with constant $q_0$ and $p_0$, is a symmetry group solution. Indeed:
\begin{align*}
\big( (x',y',\theta') \bullet z \big) (x,y,\theta)
&= z\big((x',y',\theta') (x,y,\theta)\big) \\
&= \big((x',y',\theta') (x,y,\theta)\big) \cdot \begin{pmatrix} q_0 \\ p_0 \end{pmatrix} \\
&= (x',y',\theta') \cdot \left((x,y,\theta) \cdot \begin{pmatrix} q_0 \\ p_0 \end{pmatrix} \right)\\
&= (x',y',\theta') \cdot z(x,y,\theta).
\end{align*}
The following proposition shows that all symmetry group solutions are of this form.
\end{example}

\begin{prop}\label{prop-symgpsol}
	The field $z: G \rightarrow T^*Q$ is a symmetry group solution if and only if
	\begin{equation}\label{sym-gp-sol-2}
	g \cdot z(e) = z(g) \,,
	\end{equation}
	where $e$ is the unit element of the Lie group $G$. 
\end{prop}
\begin{proof}
	Equation \eqref{sym-gp-sol-2} follows form Equation \eqref{sym-gp-sol} by choosing $h=e$. 
	
	Considering both sides of Equation \eqref{sym-gp-sol-2} as a function of $g$ and acting with $h \bullet$ we find
	\[ hg \cdot z(e) = z(hg) \,. \]
    On the other hand, considering both sides of Equation \eqref{sym-gp-sol-2} as an element of $T^*Q$ and acting with $h \cdot$ we find
	\[ hg \cdot z(e) = h \cdot z(g) \,. \]
	Hence $h \cdot z(g) = z(hg)$, which is Equation \eqref{sym-gp-sol} with $g$ and $h$ interchanged.
\end{proof}

Let $\mathfrak g$ be the Lie algebra of $G$. An abstract Lie algebra element $\xi \in \mathfrak{g}$ has two differential geometric interpretations. First, there is its representation as a left-invariant vector field $\partial_\xi \in \mathfrak{X}(G)$. It acts on functions $z$ on $G$ as
\begin{equation}\label{partial-def}
    \partial_\xi z = \frac{\d}{\d s}\Big|_{s=0} \left( \exp(s \xi) \bullet z \right) \,, 
\end{equation}
i.e.
\[ (\partial_\xi z)(g) = \frac{\d}{\d s}\Big|_{s=0} z\big( \exp(s \xi) g \big) \,. \]
In the abelian case, where $G = \R^n$ with coordinates $t_1,\ldots,t_n$, we can identify $\partial_{\xi_i} = \frac{\partial}{\partial t_i}$.
Second, there is the infinitesimal generator $V_\xi \in \mathfrak{X}(T^*Q)$ of its action on phase space. The vector field $V_\xi$ is defined by, for $z_0 \in T^* Q$,
\begin{equation}
    V_\xi( z_0) = \frac{\d}{\d s}\Big|_{s=0} \exp(s \xi) \cdot z_0 \,.
\end{equation} 
If the action of $G$ on $T^*Q$ is locally effective, then the Lie algebra of vector fields $V_\xi$ is isomorphic to $\mathfrak g$ \cite[Theorem 2.62]{olver1995equivalence}.

A symmetry group solution is characterised infinitesimally as follows:
\begin{prop}\label{prop-infinitesimal}
	The field $z: G \rightarrow T^*Q$ is a symmetry group solution if and only if, for all $g\in G$,
	\begin{equation}\label{sym-gp-infinitesimal}
	  (\partial_\xi z)(g) = V_\xi ( z(g) ) \,.
	\end{equation}
\end{prop}
\begin{proof}
	Equation \eqref{sym-gp-infinitesimal} is the infinitesimal form of
	\[ z\big( \exp(\xi) g \big) = \exp(\xi) \cdot z(g) \,, \]
	which is equivalent to Equation \eqref{sym-gp-sol}. 
\end{proof}
The motivation behind these definitions and propositions is to extend to the non-abelian case the familiar situation where one associates a time derivative $\frac{\partial}{\partial t_i}$ to each Hamiltonian vector field $V_{H_i}$ on the phase space $T^*Q$, for a family of functions $H_i$ on $T^*Q$ which are in involution, $\{H_i,H_j\}=0$. In that abelian case, for a basis $\xi_1,\ldots,\xi_n$ of the abelian Lie algebra $\mathfrak{g}$, \eqref{sym-gp-infinitesimal} reduces to 
\[\frac{\partial}{\partial t_i}z(t_1,\dots,t_n)=V_{H_i}(z(t_1,\dots,t_n))\]
and we consider the collection of these $n$ differential equations on the function $z$, giving rise to $n$ commuting time flows on the phase space $T^*Q$.

\subsection{Jet bundles over a Lie group}

As sketched in Section \ref{sec-background}, the variational principle on multi-time involves a differential form depending on configuration variables and their derivatives, i.e.\@ depending on elements of a jet bundle. 
See for example \cite{saunders1989geometry} or \cite[Section 2.3]{olver2000applications} for a detailed treatment of jet bundles.
Here, we will introduce jet bundles over Lie groups, in a slightly unusual way which allows us to understand how the Lie algebra structure affects prolongations of functions. This will help us derive multi-time Euler-Lagrange equations after we have formulated the variational principle. In the present section we discuss the jet bundle of a real function on a Lie group. It is easy to extend this to vector-valued functions (or functions into a single coordinate patch of configuration space $Q$), but for ease of presentation we restrict the discussion to real functions.

Consider a Lie group $G$ and its Lie algebra $\g$, generated by $\xi_1, \ldots, \xi_N$, with structure relations
\begin{equation}\label{strconsts}
[\xi_i,\xi_j] = \sum_k \liec_{ij}^k \xi_k \,.
\end{equation} 
We consider fields $q: G \to \R$ as sections of the trivial bundle $G \times \R \rightarrow G$ with coordinates $(g,q)$. The first jet bundle is $\cJ_1 \rightarrow G$ with $\cJ_1 = G \times \R \times \R^N$ and has coordinates $(g,q,q_1,\ldots,q_n)$. The prolongation of a smooth function $q$ to the first jet bundle is 
\[ \pr_1 q: G \to \cJ_1: g \mapsto (g, q(g), \partial_{\xi_1} q(g), \ldots, \partial_{\xi_N} q(g)) \,. \]

Starting from the second jet bundle we need to take into account the possibly non-commuting derivatives. The derivatives $q_{ij}$ and $q_{ji}$ ($i \neq j$) are not necessarily the same, but neither are they independent, because $\g$ comes with commutation relations. There are two jet bundles we could consider. Below we only present their definition for the second order jet bundle. Higher jet bundles can be constructed in an analogous way.

The first definition we present is of a jet bundle which ignores any relation between $q_{ij}$ and $q_{ji}$ ($i \neq j$), hence it is ``free'' in a similar sense as in ``free algebra''.

\begin{definition}
The \emph{free jet bundle} $\cJ_2 \rightarrow G$ has coordinates $(g, q, q_i, q_{ij})$, where both $i \leq j$ and $i > j$ are allowed, i.e.\@
\[ (g, q, q_i, q_{ij}) =  (g, q, q_i, q_{ii}, q_{ij}, q_{ji})|_{i<j} \,.\]
The fibre is thought of as the vector space spanned by $(q, q_i, q_{ij})$.
The prolongation of $q$ to the second jet bundle is 
\[ \pr_2 q: G \to \cJ_2: g \mapsto (g, q(g), \partial_{\xi_i} q(g), \partial_{\xi_j}\partial_{\xi_i} q(g)) \,. \]
\end{definition}

Taking into account the Lie algebra structure, we can reduce the free jet bundle as follows. Guided by the commutation relation
\[ \partial_{\xi_i}\partial_{\xi_j} = \partial_{\xi_j}\partial_{\xi_i} + \sum_k \liec_{ij}^k \partial_{\xi_k}, \]
we define the equivalence relation
\[ q_{ji} \sim q_{ij} + \sum_k  \liec_{ij}^k q_k, \]
and quotient the fibres of the free jet bundle by $\sim$. Since these relations reflect the Lie group structure, they will become identities for prolongations of fields. In particular, the variational principle involves prolonged fields rather than abstract bundle variables, so it will be independent of the choice of representative.

\begin{definition}
The \emph{quotiented jet bundle} is $\widetilde{\cJ}_2 \rightarrow G$, with $\widetilde \cJ_2 = G \times \R \times \R^N \times \R^{\frac{N(N+1)}{2}}$, and has coordinates $(g, q, q_i, q_{ij})|_{i \leq j}$.
The prolongation of $q$ to the second jet bundle is 
\[ \pr_2 q: G \to \cJ_2: g \mapsto (g, q(g), \partial_{\xi_i} q(g), \partial_{\xi_j}\partial_{\xi_i} q(g))|_{i \leq j} \,. \]
\end{definition}

For a function $f: \cJ_2 \rightarrow \R$ of the free jet bundle, we denote by $\widetilde f$ its projection to the quotiented jet bundle:
\[ \widetilde f(g, q, q_i, q_{ij})|_{i \leq j} = f\Big(g, q, q_i, q_{ii}, q_{ij}, q_{ij} + \sum_k  \liec_{ij}^k q_k \Big) \Big|_{i < j} \,. \]

To easily denote elements of higher free jet bundles we use index-strings: 

\begin{definition}
An \emph{index-string} is a finite sequence $\mathfrak{I} = s_1, \ldots, s_k$ where $k \in \mathbb{N}$ and $s_i \in \{1, \ldots, N\}$, where $N$ is the dimension of the Lie group $G$. An index-string defines a derivative of the field,
$q_{\mathfrak{I}} = \partial_{\xi_{s_k}} \ldots \partial_{\xi_{s_2}} \partial_{\xi_{s_1}} q$, which is a coordinate of the prolongation of $q$ to the $k$-th free jet bundle. 
\end{definition}

To easily denote elements of higher quotiented jet bundles we use multi-indices: 

\begin{definition}
A \emph{multi-index} is an element of $\mathbb{N}^N$. A multi-index $I = (i_1,\ldots,i_N) \in \mathbb{N}^N$ defines a derivative of the field, $q_I = \partial_{\xi_N}^{i_N} \ldots \partial_{\xi_1}^{i_1} q$, which is a coordinate of the prolongation of $q$ to the $k$-th quotiented jet bundle. 
\end{definition}

When no confusion is possible, we will also use a string notation for multi-indices, for example both ``$12$'' and ``$21$'' represent the multi-index $(1,1,0,\ldots,0)$. We use the notation $\emptyset$ for the empty index-string and for the corresponding multi-index $(0,\ldots,0)$.

For a function $f: \cJ_k \rightarrow \R$ of a free jet bundle we denote by $\D_i$ the total derivative 
\[ \D_i f = \sum_{\mathfrak{I}} \der{f}{q_\mathfrak{I}} \partial_{\xi_i} q_\mathfrak{I} + \partial_{\xi_i} f\,, \]
where the sum is over all index-strings $\mathfrak{I}$ and the final term is the analogue of $\der{f}{t_i}$ in the case of commuting flows, which vanishes in case $f$ does not depend explicitly on $g \in G$.
For a function $\widetilde f: \widetilde{\cJ}_k \rightarrow \R$ of a quotiented jet bundle we denote by $\widetilde{\D}_i$ the total derivative 
\[ \widetilde{\D}_i \widetilde f = \sum_{I} \der{\widetilde f}{q_I} \widetilde{\partial_{\xi_i} q_I} + \partial_{\xi_i} \widetilde f \,, \]
where the sum is over all multi-indices $I$. Recall that $\widetilde{\ }$ denotes projection onto the quotiented jet bundle, so $\widetilde{\partial_{\xi_i} q_I}$ means that we evaluate $\partial_{\xi_i} q_I$ while taking into account the commutation relations to write it as a linear combination of well-ordered derivatives.
For example, if $i=1$ and $\widetilde f = q_2$ the only nonzero term in the sum is the one with $I = (0,1,0,\ldots,0)$, i.e.\@ with $q_I = q_2$, yielding
\[ \widetilde{\D}_1 q_2 = \widetilde{ \partial_{\xi_1} q_2} = \widetilde{ q_{21} } = q_{12} + \sum_k  \liec_{12}^k q_k \,. \]
Note that for any function $f: G \times \mathcal J^2 \rightarrow \R$ there holds $\widetilde \D_i \widetilde f  = \widetilde{ \D_i f}$.

\begin{prop}\label{prop-commDandpar}
	For functions $f: \cJ_1 \rightarrow \R$ of the first free jet bundle there holds
	\begin{subequations}
	\begin{align}
	& \der{}{q} \D_i f = \D_i \der{f}{q} \,, \\
	& \der{}{q_j} \D_i f = \D_i \der{f}{q_j} + \delta_i^j \der{f}{q} \,, \\ 
	& \der{}{q_{ji}} \D_i f = \der{f}{q_j} \,, \\
	& \der{}{q_{ij}} \D_i f = \delta_i^j \der{}{q_j} \,, \\
	& \der{}{q_{jk}} \D_i f = 0 \qquad \text{if } j,k \neq i \,,
	\end{align}
	\end{subequations}
	where $\delta_i^j$ is the Kronecker delta.
\end{prop}
\begin{proof}
	We have
	\[ \D_i =  q_i \der{}{q} + \sum_k q_{ki} \der{}{q_k} + \partial_{\xi_i} \,. \]
	Since none of the coefficients contain an undifferentiated $q$, it follows that
	\[ \der{}{q} \D_i = \D_i \der{}{q} \,. \]
	Furthermore, we find
	\[ \der{}{q_j} \D_i = \D_i \der{}{q_j} + \delta_i^j \der{}{q} \]
    and
	\[\der{}{q_{ji}} \D_i = \D_i \der{}{q_{ji}} + \der{}{q_j} = \der{}{q_j} \,, \]
	where the last equality holds because we are only considering function of the first jet bundle.
	Similarly, if $j \neq i$ and $k\neq i$ we find
	$\der{}{q_{ij}} \D_i = 0$ and $\der{}{q_{jk}} \D_i = 0$.
\end{proof}

For the sake of completeness, we also state the corresponding result in the quotiented jet bundle. However, in the calculations to follow we will always commute total and partial derivatives using Proposition \ref{prop-commDandpar}, before projecting to the quotiented bundle.

\begin{prop}
	For functions $f: \widetilde{\cJ}_1 \rightarrow \R$ of the first quotiented jet bundle there holds
	\begin{subequations}
	\begin{align}
	& \der{}{q} \widetilde{\D}_i f = \widetilde{\D}_i \der{f}{q} \,, \\
	& \der{}{q_j} \widetilde{\D}_i f = \widetilde{\D}_i \der{f}{q_j} + \delta_i^j \der{f}{q} + \sum_{k > i} \liec_{ik}^j \der{f}{q_k} \,, \\ 
	& \der{}{q_{ij}} \widetilde{\D}_i f = \der{}{q_{ji}} \widetilde{\D}_i f = \der{}{q_j} \,, \\
	& \der{}{q_{jk}} \widetilde{\D}_i f = 0 \qquad \text{if } j,k \neq i \,,
	\end{align}
	\end{subequations}
	where $\delta_i^j$ is the Kronecker delta.
\end{prop}
\begin{proof}
	The proof is analogous to the proof of Proposition \ref{prop-commDandpar}, starting from the expansion
	\begin{align*}
	    \widetilde{\D}_i &=  q_i \der{}{q} + \sum_k \widetilde{q_{ki}} \der{}{q_k}  + \partial_{\xi_i}\\
	     &= q_i \der{}{q} + \sum_{k \leq i} q_{ki} \der{}{q_k} + \sum_{k > i} \left( q_{ik} + \sum_\ell \liec_{ik}^\ell q_\ell \right) \der{}{q_k}  + \partial_{\xi_i} \,. \qedhere
	\end{align*}
\end{proof}

\subsection{Variational principle for functions on a Lie group}

We are now in a position to formulate the variational principle on a Lie group and derive the corresponding multi-time Euler-Lagrange equations. In this subsection we will state the definition for general $d$-forms and derive some results which will be helpful to carry out the calculus of variations. In Sections \ref{sec-1form} and \ref{sec-2form} we will specialise the discussion to $d=1$ and $d=2$.

\begin{definition}\label{def-varpcple}
	Consider a $d$-form $\cL[q]$ on a Lie group $G$. We say that $q: G \rightarrow Q$ is \emph{critical} if for every $d$-dimensional submanifold $\gamma \subset G$, we have
	\begin{equation}\label{critical-on-G}
	\der{}{\varepsilon}\Big|_{\varepsilon = 0} \int_\gamma \cL[q + \varepsilon \eta] = 0 \,,
	\end{equation}
	for any variation $\eta$ that vanishes (along with all its derivatives) at the boundary of $\gamma$.
\end{definition}

The following characterisation of critical fields makes use of the vertical exterior derivative $\delta$ in the variational bicomplex (see for example \cite{anderson1992introduction}, \cite[Appendix A]{suris2016lagrangian}, or \cite{caudrelier2020hamiltonian}). This $\delta$ can be thought of as taking an infinitesimal variation (Gateaux derivative) in a direction yet to be specified. For example, $\delta q$ is an operator which maps a vector field to the variation of $q$ in the direction of this vector field. In the variational bicomplex, $\delta$ anti-commutes with $\d$.

\begin{lemma}\label{lemma-equivalence}
	The following are equivalent:
	\begin{enumerate}[$(i)$]
		\item $q: G \rightarrow Q$ is critical,
		\item All infinitesimal variations of the exterior derivative of $\cL$ vanish, i.e.\@ \[\delta \d \cL[q] = 0 \,.\]
	\end{enumerate}
\end{lemma}
\begin{proof}
	\begin{description}
		\item{$(i) \Rightarrow (ii)$} Let $q$ be critical and consider a $(d+1)$-dimensional oriented submanifold $D$ of $G$ with boundary $\partial D$. Then by Stokes theorem and the variational principle \eqref{critical-on-G} with $\gamma = \partial D$ there holds
		\begin{equation}\label{delta-d-L}
		     \int_D \delta \d \cL = -\int_D \d \delta \cL = -\int_\gamma \delta \cL = 0 \,.
		\end{equation} 
		Since $D$ is arbitrary, this implies that $\delta \d \cL = 0$. 
		
		\item{$(ii) \Rightarrow (i)$} Reading Equation \eqref{delta-d-L} from right to left, we see that if $\delta \d \cL = 0$ then the variational principle \eqref{critical-on-G} is satisfied for all closed $d$-dimensional submanifolds $\gamma = \partial D$. Below we argue that this implies that the variational principle holds on all $d$-dimensional submanifolds.
		
		In the variational principle it is sufficient to consider variations supported in a small neighbourhood, because using a partition of unity we can write any variation as a sum of variations with smaller supports. Hence we can assume that the manifold $\gamma$ in Equation \eqref{critical-on-G} is bounded, so we can extend it to a closed $d$-dimensional submanifold $\bar \gamma$, such that $\gamma \subset \bar \gamma$. Since we already established that the variational principle holds on closed submanifolds, it now follows that is also holds on $\gamma$.
		\qedhere
	\end{description}
\end{proof}

In the following sections, we will use Lemma \ref{lemma-equivalence} to derive the multi-time Euler-Lagrange equations in the case of 1-forms and 2-forms. Before doing so, we explore some properties of the vertical exterior derivative $\delta$.

For any function $P$ of the free jet bundle, there holds
\begin{equation}\label{delta-freejet}
 \delta P = \sum_{\mathfrak{I}} A^\mathfrak{I} \delta q_\mathfrak{I} \,,
\end{equation}
where
\[ A^\mathfrak{I} = \der{P}{q_\mathfrak{I}} \]
for every index-string $\mathfrak{I}$.
Because of the redundant nature of the set of index-strings, the $\delta q_\mathfrak{I}$ in this sum are not independent when we take into account the commutation relations. An expansion into $\delta q_I$, which are independent in the quotiented jet bundle, is obtained in the following Lemma.

\begin{lemma}\label{lemma-free2quotient}
	For any function $P: \cJ^2 \rightarrow \R$ of the free second jet bundle. The vertical exterior derivative of its projection onto the quotiented bundle reads
	\[ \delta \widetilde P = \sum_{I} B^I \delta q_I \,, \]
	where 
	\begin{equation}\label{A-tilde}
	\begin{split}
	& B^{kk} = \widetilde{ A^{kk} } \,, \\
	& B^{k \ell} = \widetilde{ A^{k \ell} } + \widetilde{ A^{\ell k} } \quad \text{where } k < \ell \,, \\
	& B^k = \widetilde{ A^k } +\sum_{m<n}\liec_{mn}^k \widetilde{ A^{nm} } \,, \\
	& B^{\emptyset}= \widetilde{ A^{\emptyset} } \,,
	\end{split}
	\end{equation}
	$A^J = \der{P}{v_J}$, and $\widetilde{\ }$ denotes the projection onto the quotiented jet. In particular, we define $\widetilde{ \delta q_\mathfrak{I}} = \delta \widetilde{q_\mathfrak{I}}$, hence $\widetilde{ \delta P } = \delta \widetilde P$.
\end{lemma}
\begin{proof}
	We start from the expansion \eqref{delta-freejet} in the free jet bundle. From the Lie bracket relations (Equation \eqref{strconsts}) it follows that
	\[ \delta q_{ji}=\delta q_{ij}+\sum_k\liec_{ij}^k\delta q_k \,. \]
	Using this relation we find the $B^I$ in terms of the $A^\mathfrak{I}$ as in Equation \eqref{A-tilde}.
\end{proof}

\section{Lagrangian 1-forms on Lie groups}
\label{sec-1form}

Let $\cL[q]$ be a 1-form on $G$, depending on the first jet of a field $q$. It is determined by its pairings $L_i[q]$ with the generators $\partial_{\xi_i} \in \mathfrak{X}(G)$ of the Lie algebra,
\[ \partial_{\xi_i} \ip \cL[q] = L_i[q] \,.\]
Once again we pose the variational principle that the action along every curve must have a critical value with respect to variations of $q:G \rightarrow Q$.
The same multi-time Euler-Lagrange equations as in the commutative case apply:
\begin{thm}\label{thm-mtEL1}
    If $\cL$ only depends on the first jet of $q$, the variational principle of Definition \ref{def-varpcple} is equivalent to the following set of multi-time Euler-Lagrange equations
	\begin{align}
	&\der{L_j}{q_j} = \der{L_i}{q_i} \,, \label{EL1-corner} \\
	&\der{L_j}{q_k} = 0  \qquad \text{ if } k \neq j \,, \label{EL1-alien} \\
	&\var{j}{L_j}{q} = 0 \,. \label{EL1-line}
	\end{align}
\end{thm}

\begin{proof}
	Let $P_{ij}[q] = \partial_{\xi_j} \ip \partial_{\xi_i} \ip \d \cL[q]$, with $i < j$. Then 
	\begin{align*}
	P_{ij} 
	&= \D_i (\partial_{\xi_j} \ip \cL) - \D_j (\partial_{\xi_i} \ip \cL) - \partial_{[\xi_i,\xi_j]} \ip \cL \\
	&= \D_i L_j - \D_j L_i - \sum_k \liec_{ij}^k L_k \,.
	\end{align*}
	We consider $P_{ij}$ as a function of the free jet bundle, i.e.\@ we allow it to contain derivatives that are not well-ordered. We have
	\[ \delta P_{ij} = \sum_{\mathfrak I} A^{\mathfrak I} \delta q_{\mathfrak I} \,, \]
	where the sum is over all index-strings and
	\[ A^{\mathfrak I} = \der{P_{ij}}{q_{\mathfrak I}} \,. \]
	Since we are on the free jet bundle, we are ignoring the commutation relations, hence the $\delta q_{\mathfrak I}$ are not independent. To remedy this we use Lemma \ref{lemma-free2quotient} and find
	\[ \delta P_{ij} = \sum_{I} B^I \delta q_I \,, \]
	where the sum is over all multi-indices and the $B^I$ are given by Equation \eqref{A-tilde}.
	In particular, using Proposition \ref{prop-commDandpar} we find
	\begin{align*}
	&B^{ii} = \widetilde{A^{ii}} = \widetilde{ \der{P_{ij}}{q_{ii}} } = \der{L_j}{q_i} \,,
	\\
	&B^{ij} = \widetilde{A^{ij}} + \widetilde{A^{ji}} = \widetilde{ \der{P_{ij}}{q_{ij}} } + \widetilde{ \der{P_{ij}}{q_{ji}} } = \der{L_i}{q_i} - \der{L_j}{q_j} \,,
	\\
	&B^{ik} = \widetilde{A^{ik}} + \widetilde{A^{ki}} = \widetilde{ \der{P_{ij}}{q_{ik}} } + \widetilde{ \der{P_{ij}}{q_{ki}} } = - \der{L_j}{q_k} \,.
	\end{align*}
	Hence $\delta P_{ij} = 0$ implies Equations \eqref{EL1-corner} and  \eqref{EL1-alien}. Furthermore, using Equations \eqref{EL1-corner}--\eqref{EL1-alien},
	\begin{align*}
    B^i = \widetilde{A^i} + \sum_{m<n}\liec_{mn}^i \widetilde{A^{nm}}
	&= \widetilde{ \der{P_{ij}}{q_i} } + \liec_{ji}^i \widetilde{ \der{P_{ij}}{q_{ij}} } \\
	&= \der{L_j}{q} - \D_j \der{L_i}{q_i} + \liec_{ij}^i \der{L_i}{q_i} + \liec_{ji}^i \der{L_i}{q_i} \\
	&= \var{j}{L_j}{q} \,,
	\end{align*}
	hence $\delta P_{ij} = 0$ also implies Equation \eqref{EL1-line}.
	
	We have shown, using Lemma \ref{lemma-equivalence}, that the variational principle implies Equations \eqref{EL1-corner}--\eqref{EL1-line}. The converse can be proved by retracing our steps and observing that the equations $B^k = 0$ ($k \neq i,j$) and $B^\emptyset = 0$ are consequences of Equations \eqref{EL1-corner}--\eqref{EL1-line} as well.
\end{proof}

An alternative proof can be given using the stepped curve approach of \cite{suris2016lagrangian}, which for 1-forms easily generalises to the Lie group setting. The stepped curve approach also applies to forms depending on higher jets, and again leads to the same multi-time Euler-Lagrange equations as in the commutative case.

\subsection{Building a 1-form from symmetries of a given Lagrangian}
\label{sec-varsym}

In this subsection we start from a mechanical Lagrangian (of Newtonian type) and build a Lagrangian 1-form describing its variational symmetry group. As is common in this context, we use $\dot{q}$ as shorthand for the time-derivative of $q$, which will be identified $\partial_{\xi_1} q$ in the multi-time setting. Suppose we are given a mechanical system on $TQ$ with Lagrangian
\begin{equation} 
\label{mechlag}
L_1(q,\dot{q}) = \frac{1}{2} |\dot{q}|^2 - U(q)
\end{equation}
and a group $G_0$ of variational symmetries of $L_1$. We allow elements of $G_0$ to be generalised symmetries, meaning that they do not necessarily act as point transformations on $Q$ and have infinitesimal generators that potentially depend on derivatives of the curve. We do require the symmetries in $G_0$ to act on the space of curves from $\R$ to $Q$. This rules out some generalised symmetries which are defined infinitesimally as a generalised vector field but cannot be integrated to a symmetry transformation \cite[Chapter 5]{olver2000applications}.
By definition, $g$ is a variational symmetry if for any $q: \R \rightarrow Q$:
\[ L_1 \left( (g\cdot q)(t) ,\frac{\d}{\d t} (g \cdot q)(t) \right) = L_1(q(t),\dot{q}(t)) + \frac{\d}{\d t} \mathcal{F}_g(q(t),\dot{q}(t)) \]
for some function $\mathcal{F}_g: TQ \rightarrow \R$.

We assume that the infinitesimal generators of all $g \in G_0$ are prolongations of vector fields of the form $W = w(q,\dot{q}) \partial_q$:
\[ \pr(W) = w(q,\dot{q}) \partial_q + \left( \der{w}{q} \dot{q} + \der{w}{\dot{q}} \ddot{q} \right) \partial_{\dot q} \,. \]
The infinitesimal characterisation of a variational symmetry reads
\begin{equation}
\label{varsym}
\pr(W) L_1(q,\dot{q}) = \frac{\d}{\d t} F_W(q,\dot{q}) \,,
\end{equation}
where the function $F_W$ is called the \emph{flux} of the variational symmetry.

In this setting, a natural choice of multi-time is $G = \R \times G_0$. The additional $\R$ represents translations in time $t$. Picking some reference time $(t_0,e) \in G$ and initial conditions $(q_0, v_0) \in TQ$ such that $q(t_0,e) = q_0$, $\dot{q}(t_0,e) = v_0$, the Euler-Lagrange equation of the Lagrangian \eqref{mechlag} defines $q(t,e)$ for all $t \in \R$. This solution can be extended to a symmetry group solution on $G$ by
\begin{equation}\label{solution}
q(t,h) = \Phi^h(q(t,e)) = h \cdot q(t,e)
\end{equation}
for $h \in G_0$ (see Proposition \ref{prop-symgpsol}). When writing \eqref{solution}, we took advantage of the fact that the $t$-flow and $\Phi^h$ commute (because $G_0$ is a symmetry group of $L_1$), hence the order in which the two flows are applied does not matter.

The infinitesimal characterisation (see Proposition \ref{prop-infinitesimal}) of a symmetry group solution is that for all $g \in G$ and $\xi \in \g$:
\begin{align*}
\ddot{q}(g) &= -U'(q(g)) \,, \\
\partial_\xi q(g) &=  W_\xi q(g) 
= w_\xi(q(g),\dot{q}(g)) \,,
\end{align*}
where $W_\xi = w_\xi \partial_q$ is the infinitesimal generator of the group action and the differential operator $\partial_\xi$ is defined in Equation \eqref{partial-def}.

Let $\xi_1= \partial_t, \xi_2,\ldots,\xi_N$ be a basis of the Lie algebra $\mathfrak{g}$ of $G$ and $w_{\xi_i}(q,\dot{q}) \partial_q$ the corresponding generalised vector fields, which are assumed to be variational symmetries of \eqref{mechlag} with fluxes $F_i(q,\dot{q})$. Following \cite{petrera2017variational} we consider for $i \geq 2$
\begin{align}
L_i(q,q_1,q_i) &= \der{L_1}{q_1} (q_i - w_{\xi_i}(q,q_1)) + F_i(q,q_1) \notag \\
&= q_1 q_i - q_1 w_{\xi_i}(q,q_1) + F_i(q,q_1) \,, \label{varsymlag}
\end{align}
where $q_i = \partial_{\xi_i}q$.
The key point here is that we do not assume that the variational symmetries commute. This is a departure from the setting of \cite{petrera2017variational} and other previous works on Lagrangian multiforms. We allow a general Lie algebra structure:
\begin{equation}
[\xi_i, \xi_j] = \sum_k \liec_{ij}^k \xi_k \,. 
\end{equation}
Given $q: G \rightarrow Q$, we define the \defn{Lagrangian 1-form $\cL[q]$} on $G$ by
\begin{equation}
\xi_i \,\ip \cL[q] = L_i(q,q_1, q_i) \,,
\end{equation}
where $L_i$ is defined by Equation \eqref{varsymlag} for $i \geq 2$ and by Equation \eqref{mechlag} for $i=1$.
Note that if all $\xi_i$ commute, then $G \cong \R^{N}$ and we can choose coordinates $t_i$ such that $\cL = \sum_{i=1}^N L_i \,\d t_i$, which is the familiar expression for a Lagrangian $1$-form, as it appears in the literature in the context of commuting symmetries.
The definition of $\cL$ is independent of the basis $\xi_1,\ldots,\xi_n$, as the following proposition shows. 
\begin{prop}
	For every $\xi = \sum_k \alpha_k \xi_k \in \mathfrak{g}$ and every $q: G \rightarrow Q$ there holds 
	\[ \xi \,\ip \cL[q] = q_1 \partial_\xi q - q_1 w_\xi(q,q_1) + \sum_k \alpha_k F_k(q,q_1) \,, \]
	where $w_\xi = \sum_k \alpha_k w_{\xi_k}$ is the characteristic of the generalised vector field on $Q$ induced by $\xi$ and $F_k$ are the fluxes of the variational symmetries $\xi_k$, with $F_1 = \cL_1$ and $w_{\xi_1} = q_1$.
\end{prop}
\begin{proof}
	We have
	\begin{align*}
	\xi \,\ip \cL =
	\sum_k \alpha_k \xi_k \,\ip \cL
	&=  \sum_k \alpha_k L_k(q,q_1, q_k) \\
	&=  \sum_k \alpha_k \left( q_1 q_k - q_1 w_{\xi_k}(q,q_1) + F_k(q,q_1) \right) \\
	&=  q_1 \partial_\xi q - \sum_k \alpha_k q_1 w_{\xi_k}(q,q_1) + \sum_k \alpha_k F_k(q,q_1) \,.
	\qedhere
	\end{align*}
\end{proof}

\begin{thm}\label{thm-sol=sol}
	A field $q: G \rightarrow Q$ is a symmetry group solution of the group action of $G$ if and only if it is critical in the sense of Definition \ref{def-varpcple} for the Lagrangian 1-form $\cL$.
\end{thm}
The proof is essentially the same as that of \cite[Prop. 7.1]{petrera2017variational}. First we prove some Lemmas.

\begin{lemma}[{\cite[Lemma 3.1]{petrera2017variational}}]
	\label{lemma-Fv}
	Let $F$ be the flux of a variational symmetry $W = w \partial_q$ of the Lagrangian \eqref{mechlag}. There holds 
	\[\der{F}{\dot q^\alpha} = \der{w^\beta}{\dot q^\alpha} \dot q^\beta \,,\]
	where the Greek upper indices denote vector components and summation over repeated indices is assumed.
\end{lemma}
\begin{proof}
	Using the chain rule we find from Equation \eqref{varsym}
	\[ -\der{U}{q^\beta}  w^\beta + \dot{q}^\beta \left( \der{w^\beta}{q^\alpha} \dot{q}^\alpha +  \der{w^\beta}{\dot{q}^\alpha} \ddot{q}^\alpha \right) = \der{F}{q^\alpha} \dot{q}^\alpha + \der{F}{\dot{q}^\alpha} \ddot{q}^\alpha \,. \]
	Since this holds for any curve $q$, the coefficients of $\ddot{q}$ must match, hence $\dot{q}^\beta \der{w^\beta}{\dot{q}^\alpha} = \der{F}{\dot{q}^\alpha}$.
\end{proof}

\begin{lemma}
	\label{lemma-Fq}
	On solutions of the Euler-Lagrange equation $\ddot{q} + U'(q)=0$ there holds
	\[ \der{F}{q^\alpha} = \frac{\d w^\alpha}{\d t} + \dot q^\beta \der{w^\beta}{q^\alpha} \,,\]
	where summation over repeated indices is assumed.
\end{lemma}
\begin{proof}
	We have
	\begin{align*}
	\der{F}{q^\alpha} &= \der{}{\dot q^\alpha} \frac{\d F}{\d t} - \frac{\d}{\d t} \der{F}{\dot q^\alpha} \\
	&= \der{}{\dot q^\alpha} \pr(W) L_1 - \frac{\d}{\d t} \left(\dot{q}^\beta \der{w^\beta}{\dot q^\alpha}\right) \\
	&= \der{}{\dot q^\alpha} \left( \dot{q}^\beta \frac{\d w^\beta}{\d t} - \der{U}{q^\beta} w^\beta \right) - \ddot{q}^\beta \der{w^\beta}{\dot q^\alpha} - \dot{q}^\beta \frac{\d}{\d t}\der{w^\beta}{\dot q^\alpha} \\
	&= \frac{\d w^\alpha}{\d t} +  \dot{q}^\beta \der{}{\dot q^\alpha} \frac{\d w^\beta}{\d t} - \der{U}{q^\beta} \der{w^\beta}{\dot q^\alpha} - \ddot{q}^\beta \der{w^\beta}{\dot q^\alpha} - \dot{q}^\beta \frac{\d}{\d t}\der{w^\beta}{\dot q^\alpha} \\
	&= \frac{\d w^\alpha}{\d t} + \dot{q}^\beta \der{w^\beta}{q^\alpha} - \left(\ddot{q}^\beta + \der{U}{q^\beta} \right) \der{w^\beta}{\dot q^\alpha} \,. \qedhere
	\end{align*}
\end{proof}

\begin{proof}[Proof of Theorem \ref{thm-sol=sol}]
	Critical fields $q: G \rightarrow Q$ are characterised by the multi-time Euler-Lagrange equations \eqref{EL1-corner}--\eqref{EL1-line}. In particular, Equation \eqref{EL1-line} with $j = 1$ yields
	\[ q_{11} + U'(q) = 0 \]
	and Equation \eqref{EL1-alien} with $k=1$ yields
	\[  q_j^\alpha - w_{\xi_j}^\alpha - q_1^\beta \der{w_{\xi_j}^\alpha}{q_1^\beta} + \der{F}{q_1^\alpha} = 0 \,, \]
	which by Lemma \ref{lemma-Fv} is equivalent to
	\[ q_j - w_{\xi_j} = 0 \,. \]
	Equation \eqref{EL1-corner} is equivalent to the trivial equation $\dot{q} = \dot{q}$. Equation \eqref{EL1-line} with $j \neq 1$ is a differential consequence of the previous equations, as the following computation shows. We have
	\begin{align*}
	\eqref{EL1-line} &\Leftrightarrow  - q_1^\beta \der{w_{\xi_j}^\beta}{q^\alpha} + \der{F_j}{q^\alpha} - q_{1j}^\alpha = 0 
	\xLeftrightarrow{\text{Lemma \ref{lemma-Fq}}} \partial_1 w_{\xi_j}^\alpha - q_{1j}^\alpha = 0 \,,
	\end{align*}
	which is equivalent to $\partial_1 ( w_{\xi_j}^\alpha - q_j^\alpha )$, hence it is a consequence of Equation \eqref{EL1-alien} with $k=1$.
\end{proof}

\subsection{Building a 1-form from Hamiltonians}

Suppose we have a Hamiltonian of Newtonian type, $H(q,p) = \frac{1}{2} |p|^2 + U(q)$, with a number of symmetries defined by Hamiltonians $H_2,\ldots,H_N$. Let $X_{H_1},  \ldots, X_{H_N}$ be the corresponding Hamiltonian vector fields on $T^*Q$.

\begin{prop}\label{prop-linear-poisson}
	If there exist constants $\liec_{ij}^k$ such that
	\begin{equation}\label{linear-poisson}
	\{ H_i, H_j \} = \sum_k \liec_{ij}^k H_k \,,
	\end{equation}	
	then the vector space spanned by the Hamiltonian vector fields $X_{H_1},\ldots,X_{H_N}$ is a Lie algebra with Lie bracket given by the commutator. 
\end{prop}
\begin{proof}
	We have
	\begin{align} 
	[X_{H_i},X_{H_j}] f 
	&= \{ H_i , \{ H_j, f \} \} - \{ H_j , \{ H_i, f \} \} \notag\\
	&= \{ \{ H_i, H_j \}, f \} \notag\\
	&= \sum_k \liec_{ij}^k \{ H_k, f \} \,. \label{Lie-Poisson}
	\end{align}
	Hence $[X_{H_i},X_{H_j}] = \sum_k \liec_{ij}^k X_{H_k}$.
\end{proof}

In this case, we can take as our multi-time the universal covering Lie group $G$ of the Lie algebra $\g$. In other words, we take the group of symmetries generated by $H_1, \ldots, H_N$ as multi-time. 

The role played in the commuting case by the time derivatives $\partial_{t_i}$ is now played by a basis of left-invariant vector fields $\partial_{\xi_i}$ on $G$. They satisfy the same Lie-algebraic relations as the $X_{H_i}$, but we choose not to identify them. This is to emphasise the conceptual difference between the vector fields $\partial_{\xi_i}$ on multi-time $G$ and the vector fields $X_{H_i}$ on the phase space $T^*Q$. The Hamiltonian vector fields $X_{H_i}$ play the role of the symmetry generators $V_{\xi_i}$ in Proposition \ref{prop-infinitesimal}.

If the Poisson relations fail to be linear, i.e.\@ if Equation \eqref{linear-poisson} does not hold for any constants $\liec_{ij}^k$, then we need a different approach to find a Lie algebra and Lie group on which to formulate the variational principle. This case will be handled in Section \ref{sec-nonlin-poisson}. First we will discuss the construction of a suitable 1-form on $G$ in case Equation \eqref{linear-poisson} does hold.

In the context of commuting flows, the relation between Lagrangian 1-forms and Hamiltonian structures is well understood \cite{suris2013variational,vermeeren2021hamiltonian}, at least if they are of Newtonian form. We use the same construction here and define the 1-form $\cL[q]$ by
\begin{equation}\label{lag-from-ham}
\partial_{\xi_i} \ip \cL[q] = L_i[q] := q_1 q_i - H_i(q,q_1) \,,
\end{equation}
where we identify $p = q_1$. Note that $L_1 = \frac{1}{2} q_1 - U(q)$.

\begin{prop}\label{prop:EL=Ham}
    The multi-time Euler-Lagrange equations for $\cL$, as defined in Equation \eqref{lag-from-ham}, are equivalent to the set of canonical Hamiltonian equations for $H_1,\ldots,H_N$, under the identification $p = q_1.$
\end{prop}
\begin{proof}
The multi-time Euler-Lagrange equations of type \eqref{EL1-corner} are trivially satisfied.
The multi-time Euler-Lagrange equations of type \eqref{EL1-alien} yield
\[ q_i = \der{H_i(q,q_1)}{q_1} \]
for $i > 0$ (and nothing for $L_1$).
Finally, the multi-time Euler-Lagrange equations of type \eqref{EL1-line} yield
\[ q_{1i} = - \der{H_i(q,q_1)}{q} \]
for $i > 0$ and 
\[ q_{11} = -\der{V(q)}{q} \,. \qedhere\]
\end{proof}

\begin{prop}\label{prop:double0}
The exterior derivative $\d \cL$, where $\cL$ is as in Equation \eqref{lag-from-ham}, attains a double zero on solutions to the multi-time Euler-Lagrange equations.
\end{prop}
\begin{proof}
    First, observe that
    \begin{align*}
        \partial_{[\xi_i,\xi_j]} \ip \cL 
        &= \sum_k \liec_{ij}^k L_k \\
        &= q_1 \sum_k \liec_{ij}^k q_k -  \sum_k \liec_{ij}^k H_k(q,q_1) \\
        &= q_1 [\partial_{\xi_i},\partial_{\xi_j}]q - \{ H_i, H_j \}(q,q_1) \,.
    \end{align*}
    Hence, by linearity, for all $\xi,\nu \in \g$ there holds
    \[ \partial_{[\xi,\nu]} \ip \cL = q_1 [\partial_\xi,\partial_\nu]q - \{ H_\xi, H_\nu \}(q,q_1) \,. \]
    Now, identifying $p = q_1$, we find
    \begin{align*}
        \partial_\nu \ip \partial_\xi \ip \d \cL
        &= \partial_\xi L_\nu - \partial_\nu L_\xi - \partial_{[\xi,\nu]} \ip \cL  \\
        &= \partial_\xi p \partial_\nu q - \partial_\nu p \partial_\xi q - \partial_\xi H_\nu + \partial_\nu H_\xi + \{ H_\xi, H_\nu \} \\
        &= \partial_\xi p \partial_\nu q - \partial_\nu p \partial_\xi q - \der{H_\nu}{q} \partial_\xi q - \der{H_\nu}{p} \partial_\xi p + \der{H_\xi}{q} \partial_\nu q - \der{H_\xi}{p} \partial_\nu p + \{ H_\xi, H_\nu \} \\
        &= \left( \partial_\xi p + \der{H_\xi}{q} \right) \left( \partial_\nu q - \der{H_\nu}{p} \right) 
        -\left( \partial_\nu p + \der{H_\nu}{q} \right) \left( \partial_\xi q - \der{H_\xi}{p} \right) \,,
    \end{align*}
    where the last equality makes use of the canonical form of the Poisson bracket,
    $ \{f,g\} = \der{f}{p}\der{g}{q} - \der{g}{p} \der{f}{q}. $
\end{proof}

\subsubsection{Nonlinear Poisson relations}
\label{sec-nonlin-poisson}
If the Poisson relations are not linear, then (some of) the coefficients of $\{ H_k, f \}$ in Equation \eqref{Lie-Poisson} will depend on the $H_i$ instead of being constant. This would mean that the commutators between the vector fields depend on the values of the Hamiltonians, which means that the vector fields themselves do not constitute a Lie algebra. However, even in this case, we can find a Lie algebra underlying the system.

Consider the commutative algebra $\cF$ generated by the functions $H_1, \ldots, H_N: T^*Q \rightarrow \R$, i.e.\@ the algebra of functions of $H_1, \ldots, H_N$. If this algebra is closed under the canonical Poisson bracket, then it is an example of a \emph{function group} in the sense of Lie \cite{lie1890theorie}.\footnote{For all examples in this work it would be sufficient to consider the subalgebra of this function group consisting  of polynomials in $H_1, \ldots, H_N$.}
Such function groups have been considered in geometric mechanics \cite{weinstein1983local,hermann1983geometric} and control theory \cite{vanderschaft1985controlled}.
The canonical Poisson bracket on $T^*Q$ turns $\cF$ into a Lie algebra
\[ \h = (\cF, \{\cdot,\cdot\}) \,.\]
Even though structure constants in the sense of Equation \eqref{linear-poisson} do not exist in this case, one can define structure constants on $\h$ by choosing a basis $\{ f_\alpha \}$ of $\h$. (Every vector space has a basis if the axiom of choice is assumed.) We then take $\liec_{\alpha \beta}^\gamma$ such that
\[ 	\{ f_\alpha, f_\beta \} = \sum_\gamma \liec_{\alpha \beta}^\gamma f_\gamma \,. \]
Even though the index set that $\gamma$ belongs to is infinite, the above sum will only contain a finite number of nonzero terms, because every vector can be written as a linear combination of finitely many basis elements.

Using the correspondence between functions $f \in \cF$ and their canonical Hamiltonian vector fields $X_f$, we can see that $\h$, modulo additive constants, is isomorphic to the Lie algebra
\[ ( \{ X_f \mid f \in \cF \}, [\cdot,\cdot] ) \,, \]
where $[\cdot,\cdot]$ denotes the commutator of vector fields on $T^*Q$. Note that the finite set of vector fields $\{X_{H_1}, \ldots, X_{H_N}\}$ does not necessarily span a Lie algebra. On the other hand, the infinite-dimensional $\h$ is always a Lie algebra, regardless of what the Poisson relations look like.

We would like to define a 1-form $\cL$ and impose the variational principle on a Lie group that has $\h$ as its Lie algebra. However, since Lie's (converse) third theorem does not hold in an infinite-dimensional setting \cite{vanest1964nonenlarbile}, the existence of such a Lie group cannot be guaranteed. At first sight this seems to be an insurmountable obstruction, since we need a Lie group (or at least a manifold) to formulate the variational principle of Definition \ref{def-varpcple}. However, the characterisation $\delta \d \cL = 0$ is formulated within a single tangent space of multi-time, so it can be considered without referring to the manifold structure. In particular, we can impose it on the Lie algebra $\h$ with no requirement for a Lie group structure. Lemma \ref{lemma-equivalence} shows the characterisation $\delta \d \cL = 0$ to be equivalent to the variational principle whenever a Lie group exists, and the formal calculations will be the same whether or not a Lie group associated to $\h$ exists. Below we provide some details of this construction.

As multi-time we consider a subgroup of the group of symplectic diffeomorphisms on Q, given by
\begin{equation}
G_\cF = \{ \exp(X_{f_1}) \circ \ldots \circ \exp(X_{f_m}) \mid m \in \N, f_1,\ldots,f_m \in \cF \} \,,    
\end{equation}
where $\exp$ denotes the flow of the vector field over one unit time. This is the group of all transformations which can be obtained as concatenation of finitely many flows of Hamiltonian vector fields with Hamiltonian function in $\cF$.
We do not claim that $G_\cF$ has a manifold structure. For each pair $(g,f) \in G_\cF \times \cF$ we consider the smooth one-parameter subgroup $\{ \exp(t X_f) \circ g \mid t \in \R \}$. This allows us to define the derivative of a field $q : G_\cF \rightarrow Q$ with respect to $f \in \cF$:
\begin{equation}
(\partial_f q)(g) = \frac{\d}{\d t} q( \exp(t X_f) \circ g) \big|_{t=0} \,,     
\end{equation}
where we assume that $q$ is smooth in the sense that all such derivatives exist. We check that this definition leads to the usual relation between commutators and Poisson brackets:

\begin{prop}\label{prop-commutator}
    There holds $[\partial_{f_1}, \partial_{f_2}] = \partial_{\{f_1,f_2\}}$.
\end{prop}
\begin{proof}
    For $i=1,2$ we denote by $\Phi_i^t = \exp(t X_{f_i})$ the flow of $X_{f_i}$ over time $t$.
    For any $q : G_\cF \rightarrow Q$ and $g \in G_\cF$ we have
    \begin{align*}
[ \partial_{f_1}, \partial_{f_2} ]q(g)
        &= \frac{\d}{\d s} \frac{\d}{\d t} \left( q \left( \Phi_2^s \circ \Phi_1^t \circ g \right) - q \left( \Phi_1^t \circ \Phi_2^s \circ g \right) \right) \Big|_{s=t=0} \\
        &= \frac{\d}{\d s} \frac{\d}{\d t} \left( q \left( \Phi_2^s \circ \Phi_1^t \circ \Phi_2^{-s} \circ \Phi_1^{-t} \circ y \right) - q(y) \right) \Big|_{s=t=0} \,,
    \end{align*}
    where $y = \Phi_1^t \circ \Phi_2^s \circ g$. Denote by $\Phi_{\{f_1,f_2\}}$ the flow of the commutator of $X_{f_1}$ and $X_{f_2}$. We have that
    \[ \Phi_2^s \circ \Phi_1^t \circ \Phi_2^{-s} \circ \Phi_1^{-t} = \Phi_{\{f_1,f_2\}}^{st} + O(s^2+t^2) \,. \]
    Hence we find
    \begin{align*}
        [ \partial_{f_1}, \partial_{f_2} ] q(g)
        &= \frac{\d}{\d s} \frac{\d}{\d t} \left( q \left( \Phi_{\{f_1,f_2\}}^{st} \circ y \right) - q(y) \right) \Big|_{s=t=0} \\
        &= \frac{\d}{\d s} \frac{\d}{\d t} \left( st \partial_{\{f_1,f_2\}} q(y) + O(s^2t^2) \right) \Big|_{s=t=0}
        = \partial_{\{f_1,f_2\}} q(g) \,. \qedhere
    \end{align*}
\end{proof}

Given a field $q : G_\cF \rightarrow Q$, we define a 1-tensor $\cL[q]$ by the condition that for all $f \in \cF$ there holds
\[ f \ip \cL[q] = q_1 \cdot \partial_f q - f(q,q_1) \,, \]
where $q_1 = \partial_{H_1} q$. We call $\cL[q]$ a 1-tensor, not a 1-form, because we do not assume a manifold structure on $G_\cF$. Nevertheless, we can consider a formal exterior derivative $\d \cL[q]$ defined by 
\[ f_2 \ip f_1 \ip \d \cL[q] = \partial_{f_1} (f_2 \ip \cL[q]) - \partial_{f_2} (f_1 \ip \cL[q]) - \{f_1,f_2\} \ip \cL[q] \,. \]
In this context we use $\delta \d \cL[q] = 0$ as definition of critical fields, instead of the variational principle of Definition \ref{def-varpcple}. The operator $\delta$ is a Gateaux derivative (in a direction to be specified) of tensors. It can be calculated coefficient-wise:
\[ f_2 \ip f_1 \ip (\delta \d \cL[q]) = \delta( f_2 \ip f_1 \ip \d \cL[q] ) \,. \]

The calculation of $\d \cL$ formally takes the same form as before in the proof of Proposition \ref{prop:double0}. Using Proposition \ref{prop-commutator}, we find
\[ f_2 \ip f_1 \ip \d \cL[q] = \left( \partial_{f_1} q_1 + \der{f_1}{q} \right) \left( \partial_{f_2} q - \der{f_2}{q_1} \right) - \left( \partial_{f_2} q_1 + \der{f_2}{q} \right) \left( \partial_{f_1} q - \der{f_1}{q_1} \right) \,. \]
In particular, $f_2 \ip f_1 \ip \d \cL$ has a double zero on solutions to the canonical Hamiltonian equations for $f_1$ and $f_2$, hence
$\delta \d \cL = 0$ is equivalent to the system of all Hamiltonian equations for Hamiltonian functions in $\cF$.

Since the multi-time Euler-Lagrange equations \eqref{EL1-corner}--\eqref{EL1-line} are equivalent to the condition $\delta \d \cL = 0$, they will also be equivalent to the system of Hamiltonian equations. Indeed, we find
\begin{align*}
    &0 = \der{}{q_1}(f \ip \cL) = \partial_f q - \der{f}{q_1} \qquad (f \neq H_1) \,, \\
    &0 = \der{}{q}(f \ip \cL) - \partial_f \der{}{(\partial_f q)}(f \ip \cL) = - \partial_f q_1 - \der{f}{q} \,,
\end{align*}
while
\[
    \der{}{(\partial_{f_1} q)}(f_1 \ip \cL) = \der{}{(\partial_{f_2} q)}(f_2 \ip \cL)
\]
is trivially satisfied.

The above construction can be thought of as a generalisation of the procedure described in \cite{daboul1993hydrogen} to obtain an infinite-dimensional Lie algebra for the Kepler problem. As we will see below, this particular problem actually admits a finite dimensional Lie group, because the Poisson relations can be linearised by rescaling the Runge-Lenz vector.

\subsection{Exterior derivative and Poisson bracket} 

We know from Lemma \ref{lemma-equivalence} that every critical field $q$ satisfies $\delta \d \cL = 0$.
In the case of commuting flows, it has been shown that $\d \cL = 0$ implies that the corresponding Hamiltonian functions are in involution \cite{suris2013variational,vermeeren2021hamiltonian}. In this section we generalise this property to the Lie group setting describing non-commuting flows.

For basis elements $\xi_i,\xi_j \in \mathfrak g$ (or $\xi_i,\xi_j \in \h$) we have
\begin{align*}
\partial_{\xi_j} \ip \partial_{\xi_i} \ip \d \cL[q]
&= \partial_{\xi_i} L_j(q, q_1, q_j) - \partial_{\xi_j} L_i(q,q_1, q_i) - \partial_{[\xi_i,\xi_j]} \ip \cL[q] \\
&= \partial_{\xi_i} L_j(q,q_1,q_j) - \partial_{\xi_j} L_i(q,q_1,q_i) - \sum_k \liec_{ij}^k L_k(q,q_1,q_k) \,.
\end{align*}
A basis-independent form of this expression is obtained in terms of $L_\xi =  \partial_{\xi} \ip \cL[q]$ :
\[ \partial_{\nu} \ip \partial_{\xi} \ip \d \cL[q] = \partial_{\xi} L_{\nu}(q,q_1, \partial_\nu q) - 
\partial_{\nu} L_{\xi}(q,q_1,\partial_\xi q) - L_{[\xi,\nu]}(q,q_1,\partial_{[\xi,\nu]} q) \,. \]
Hence if $\d \cL = 0$ (as is typically the case on solutions) then 
\begin{equation}
\label{L-poisson}
L_{[\xi,\nu]}(q,q_1,\partial_{[\xi,\nu]} q) = \partial_{\xi} L_{\nu}(q,q_1, \partial_\nu q) - 
\partial_{\nu} L_{\xi}(q,q_1,\partial_\xi q) \,.
\end{equation}
As we will argue below, Equation \eqref{L-poisson} is the Lagrangian form of the fundamental relation between the Poisson bracket of two Hamiltonian functions and the commutator of the corresponding vector fields.

Regardless of how a Lagrangian 1-form $\cL$ on a Lie group was constructed, we can find Hamiltonian functions corresponding to $\cL$ if $L_1$ is of a suitable form. To do this, we generalise the construction from \cite{suris2013variational,vermeeren2021hamiltonian} to our setting. Assume that $\cL$ is of the form
\begin{equation}\label{mechlagpluri}
\begin{split}
    &\partial_{\xi_1} \ip \cL = L_1(q,q_1) = \frac{1}{2} q_1^2 - V \,,\\
    &\partial_{\xi_i} \ip \cL = L_i(q,q_1,q_i) \,,
\end{split}
\end{equation}
and produces multi-time Euler-Lagrange equations of the form
\begin{align*}
    &q_{11} = f(q,q_1) \,, \\
    &q_i = w_{\xi_i}(q,q_1) \,.
\end{align*}
The Hamiltonian function associated to $L_\xi$, for $\xi = \sum_k \alpha_k \xi_k$, is
\[ H_\xi(q,p) = p w_\xi(q,p) - L_\xi(q,p,w_\xi(q,p)) \,, \]
where $w_\xi = \sum_k \alpha_k w_k$. Note that on solutions, we have $w_\xi(q,p) = \partial_\xi q$.
If the Lagrangian 1-form is constructed from variational symmetries as in Section \ref{sec-varsym}, the Hamiltonian can also be written as
\[ H_\xi(q,p) = p w_\xi(q,p) - \sum_k \alpha_k F_k(q,p) \,.\]

\begin{prop}\label{prop:ham=EL}
    The canonical Hamilton equations for $H_\xi$ are equivalent to the Euler-Lagrange equations, i.e.\@ to the lifted action of $w_\xi \partial_q$ on $(q,q_1)$, under the identification $p = q_1$.
\end{prop}
Note that this can be thought of as an inverse statement to Proposition \ref{prop:EL=Ham}. Accordingly, the proof will be quite similar, but we find it instructive to include it. 
\begin{proof}[Proof of Proposition \ref{prop:ham=EL}]
If $\xi = \xi_1$ this is nothing but the Legendre transform. Below we prove the result for $\xi$ an arbitrary linear combination of $\xi_2,\ldots,\xi_N$. Then the general claim follows from linearity. We have
\begin{align*}
\der{H_\xi}{p} &= w_\xi +  p \der{w_\xi}{p} - \der{L_\xi}{q_1} - \der{L_\xi}{q_\xi} \der{w_\xi}{p} \\
&= w_\xi + \left( p - \der{L_\xi}{q_\xi} \right) \der{w_\xi}{p} - \der{L_\xi}{q_1}
\end{align*}
and 
\begin{align*}
\der{H_\xi}{q} &= p \der{w_\xi}{q} - \der{L_\xi}{q} - \der{L_\xi}{q_\xi} \der{w_\xi}{q} \\
&= \left( p - \der{L_\xi}{q_\xi} \right) \der{w_\xi}{q} - \der{L_\xi}{q} \,.
\end{align*}
On solutions to the multi-time Euler-Lagrange equations this gives
\[ \der{H_\xi}{p} = w_\xi \qquad \text{and} \qquad \der{H_\xi}{q} = - \partial_\xi p \,. \qedhere \]
\end{proof}

Written in terms of the Hamiltonians, Equation \eqref{L-poisson} becomes
\[
p \partial_{[\xi,\nu]} q - H_{[\xi,\nu]}(q,p) = \partial_{\xi} \left( p \partial_{\nu} q - H_\nu(q,p) \right) - \partial_{\nu} \left( p \partial_{\xi} q - H_\xi(q,p) \right) \,,
\]
which simplifies to
\begin{align*}
- H_{[\xi,\nu]}(q,p) &= (\partial_{\xi} p) (\partial_{\nu} q) - \partial_{\xi} H_\nu(q,p) - (\partial_{\nu} p) (\partial_{\xi} q) + \partial_{\nu} H_\xi(q,p) \\
&= - \der{H_\xi}{q} \der{H_\nu}{p} - \{ H_\xi, H_\nu \} + \der{H_\nu}{q}\der{H_\xi}{p} + \{ H_\nu, H_\xi\} \\
&= - \{ H_\xi, H_\nu \} \,.
\end{align*}
Hence Equation \eqref{L-poisson} is a Lagrangian version of the Poisson relations between the integrals of the system \eqref{mechlag}. In summary, we have:
\begin{thm}
Let $\cL$ be of the form \eqref{mechlagpluri} and $H_\xi$ the corresponding Hamiltonian functions. The following are equivalent:
\begin{enumerate}[(i)]
    \item $\d \cL = 0$ on solutions,
    \item Equation \eqref{L-poisson} holds: $L_{[\xi,\nu]} = \partial_{\xi} L_{\nu} - 
\partial_{\nu} L_{\xi}$,
    \item $H_{[\xi,\nu]} = \{ H_\xi, H_\nu \}$.
\end{enumerate}
\end{thm}

\subsection{Example: Kepler problem}

The Kepler problem, governed by the Hamiltonian 
\[ H(q,p) = \frac{1}{2} p^2 - \frac{1}{|q|}\]
on $T^* \R^3$, is superintegrable. Rather than the obvious $SO(3)$ rotational symmetry, it actually possesses a symmetry group isomorphic to $SO(4)$. The conserved quantities are the angular momentum
$\ell = q \times p$ and the Runge-Lenz vector
\[ A = p \times \ell - \frac{q}{|q|} \,, \]
see for example \cite{goldstein1980classical,guillemin1990variations}

A Lax pair for the Kepler problem was proposed in \cite{antonowicz1992keplerlax} as a restriction of 
KdV flows (cf.\@ also \cite{torrielli2016classint} for an explicit form of this Lax representation). 
However, this Lax pair does not seem suitable for deriving the integrals of the Kepler problem. A simpler, 
but at the same time more powerful Lax representation is given as follows. Define matrices $\boldsymbol{L}_\alpha$ and $\boldsymbol{M}$ by 
\begin{align}\label{KeplerLax} 
    \boldsymbol{L}_\alpha=\left(\begin{array}{cc} q^T \alpha p & p^T \alpha p \\ 
    -q^T \alpha q & -p^T \alpha q \end{array}\right) \,, \qquad 
    \boldsymbol{M}= \left(\begin{array}{cc} 0 & |q|^{-3} \\ 
    -1  & 0 \end{array}\right) \,, 
\end{align}
where $\alpha$ is an arbitrary $3\times3$ matrix, which plays the role of the spectral parameter. Note that $\boldsymbol{M}$ does not depend on this matrix spectral parameter. An elementary calculation shows that:
\begin{prop}
The equations of motion of the Kepler problem, 
\[
\dot{q}=p\,, \qquad \dot{p}=-\frac{1}{|q|^3} \,, 
\]
follow from the Lax equation $\dot{\boldsymbol{L}}_\alpha=[ \boldsymbol{M},\boldsymbol{L}_\alpha]$. 
\end{prop}
It follows from the Lax equation that any expression of the form 
\[ \tr(\boldsymbol{L}_\alpha \boldsymbol{L}_\beta\boldsymbol{L}_\gamma \cdots )\] 
is an integral of the Kepler problem, for arbitrary choices of 3$\times$3 matrices $\alpha, \beta, \gamma, \cdots$. In particular, $\textrm{tr}(\boldsymbol{L}_\alpha)$, with $\alpha$ an arbitrary skew-symmetric matrix, yields the angular momentum vector $\ell$ as integral. Curiously, the Hamiltonian $H$ follows from the quantity $\tr(\boldsymbol{M}\boldsymbol{L}_\alpha)$ with $\alpha=\boldsymbol{1}$ taken to be the identity matrix. In fact, this quantity is by itself not an integral of the motion, but instead we have
\[ \frac{\d}{\d t}\tr(\boldsymbol{M}\boldsymbol{L}_\alpha)=\frac{\d}{\d t}\left(-\frac{3}{|q|}  \right) \,,     \] 
where the right-hand side stems from the derivative of the matrix $\boldsymbol{M}$. Thus, we can deduce that 
\[ \textrm{tr}(\boldsymbol{M}\boldsymbol{L}_\alpha)+\frac{3}{|q|}=-2H \]
is an integral, which is the Hamiltonian up to a factor. It remains an open problem if the Runge-Lenz vector $A$ arises from the Lax representation \eqref{KeplerLax}.

Using the same construction as for Hamiltonians in involution \cite{vermeeren2021hamiltonian}, we can construct a Lagrangian 1-form from these conserved quantities:
\begin{align*}
    & L_1 = \frac{1}{2} |q_1|^2 + \frac{1}{|q|} \,, \\
    & L_2 = q_1 \cdot q_2 - \ell_x = q_1 \cdot q_2 - (q \times q_1) \cdot \hat{x} \,, \\
    & L_3 = q_1 \cdot q_3 - \ell_y = q_1 \cdot q_3 - (q \times q_1) \cdot \hat{y} \,, \\
    & L_4 = q_1 \cdot q_4 - \ell_z = q_1 \cdot q_4 - (q \times q_1) \cdot \hat{z} \,, \\
    & L_5 = q_1 \cdot q_5 - A_x = q_1 \cdot q_5 - |q_1|^2 (q \cdot \hat{x}) + (q_1 \cdot \hat{x})(q_1 \cdot q) + \frac{q \cdot \hat{x}}{|q|} \,, \\
    & L_6 = q_1 \cdot q_6 - A_y = q_1 \cdot q_6 - |q_1|^2 (q \cdot \hat{y}) + (q_1 \cdot \hat{y})(q_1 \cdot q) + \frac{q \cdot \hat{y}}{|q|} \,, \\
    & L_7 = q_1 \cdot q_7 - A_z = q_1 \cdot q_7 - |q_1|^2 (q \cdot \hat{z}) + (q_1 \cdot \hat{z})(q_1 \cdot q) + \frac{q \cdot \hat{z}}{|q|} \,,
\end{align*}
where $A = (A_x,A_y,A_z)$, $\ell = (\ell_x,\ell_y,\ell_z)$ and $\{\hat{x},\hat{y},\hat{z}\}$ is the Cartesian orthonormal basis.
The multi-time Euler-Lagrange equations are given by $Q_i=0$ for $1\leq i\leq 7$, where
\begin{align*}
    & Q_1=q_{11} + \frac{q}{|q|^3} \,, \\
    & Q_2=q_2 - \hat{x} \times q \,, \\
    & Q_3=q_3 - \hat{y} \times q \,, \\
    & Q_4=q_4 - \hat{z} \times q \,, \\
    & Q_5=q_5 - 2 (q \cdot \hat{x}) q_1 + (q_1 \cdot q) \hat{x} + (q_1 \cdot \hat{x}) q \,, \\
    & Q_6=q_6 - 2 (q \cdot \hat{y}) q_1 + (q_1 \cdot q) \hat{y} + (q_1 \cdot \hat{y}) q \,, \\
    & Q_7=q_7 - 2 (q \cdot \hat{z}) q_1 + (q_1 \cdot q) \hat{z} + (q_1 \cdot \hat{z}) q \,.
\end{align*}
By cross-differentiating we can verify that $\partial_{\xi_2},\ldots,\partial_{\xi_7}$ (and hence $\xi_2,\ldots,\xi_7$) satisfy the following Lie algebra relations:
\[ \begin{tabular}{r|ccc:ccc}
    $[\xi_i,\xi_j]$ & $j=2$ & 3 & 4 & 5 & 6 & 7 \\\hline
    $i=2$ & 0 & $-\xi_4$ & $\xi_3$ & 0 & $-\xi_7$ & $\xi_6$ \\
    3 &  & 0 & $-\xi_2$ & $\xi_7$ & 0 & $-\xi_5$ \\
    4 &  &  & 0 & $-\xi_6$ & $\xi_5$ & 0 \\\hdashline
    5 &  &  &  &  0 & $2H \xi_4 + 2 \ell_z \xi_1$ & $-2H \xi_3 - 2 \ell_y \xi_1$ \\
    6 &  &  &  &  &  0 & $2H \xi_2 + 2 \ell_x \xi_1$ \\
    7 &  &  &  &  &  & 0 \\
\end{tabular}\]
If we replace $\xi_5,\xi_6,\xi_7$ by 
\begin{align*}
    \nu_5 &= \frac{\xi_5}{\sqrt{-2H}} + \frac{A_x \xi_1}{\sqrt{-2H}} \,, \\
    \nu_6 &= \frac{\xi_6}{\sqrt{-2H}} + \frac{A_y \xi_1}{\sqrt{-2H}} \,, \\
    \nu_7 &= \frac{\xi_7}{\sqrt{-2H}} + \frac{A_z \xi_1}{\sqrt{-2H}} \,,
\end{align*}
we can recover from this table the standard Lie algebra relations of $\mathfrak{so}(4)$.
Hence it is natural to define the Lagrangian 1-form $\cL$ on $\R \times SO(4)$ by $\partial_{\xi_i} \ip \cL = L_i$.

Using the properties of the triple product, the coefficients $P_{ij} = \partial_{\xi_j} \ip \partial_{\xi_i} \ip \d \cL$
are obtained by an elementary calculation. We find for example
\begin{align*}
P_{23} &= \partial_{\xi_2} L_3 - \partial_{\xi_3} L_2 - [\partial_{\xi_2},\partial_{\xi_3}] \ip L \\
&= q_{12} (q_3 - \hat{y} \times q) + q_1 (q_{32} - \hat{y} \times q_2) \\
&\qquad - q_{13} (q_2 - \hat{x} \times q) - q_1 (q_{23} - \hat{x} \times q_3) \\
&\qquad + q_1 (q_{23} - q_{32} - \hat{z} \times q) \\
&= (q_{12} - \hat{x}\times q_1) (q_3 - \hat{y} \times q) - (q_{13} - \hat{y} \times q_1) (q_2 - \hat{x} \times q) \,,
\end{align*}
which is a double zero on solutions. In general we find that for $2\leq i,j\leq 7$,
\begin{align*}
P_{1i} &= \partial_{\xi_1} L_i - \partial_{\xi_i} L_1 - [\partial_{\xi_1},\partial_{\xi_i}] \ip L \\
&= Q_1Q_i
\end{align*}
and 
\begin{align*}
P_{ij} &= \partial_{\xi_i} L_j - \partial_{\xi_j} L_i - [\partial_{\xi_i},\partial_{\xi_j}] \ip L \\
&= Q_j\partial_{\xi_1}Q_i-Q_i\partial_{\xi_1}Q_j \,,
\end{align*}
so all $P_{ij}$ have a double zero on solutions.

On the Hamiltonian side one can check that indeed $H$ is in involution with each of the components of $\ell$ and $A$. In addition, we have the following relations
\begin{align*}
    & \{ \ell_i, \ell_j \} = -\epsilon_{ijk} \ell_k \,, \\
    & \{ A_i, \ell_j \} = -\epsilon_{ijk} A_k \,, \\
    & \{ A_i, A_j \} =  2 \epsilon_{ijk} H \ell_k \,, \\
    & \{ H, \cdot \} = 0 \,,
\end{align*}
where $\epsilon_{ijk}$ is the totally anti-symmetric tensor with $\epsilon_{ijk} = 1$  (see e.g.\@ \cite[Sec.\@ 9--7]{goldstein1980classical}, but beware that we use the opposite sign convention for the Poisson bracket). These Poisson brackets reflect the commutation relations between the corresponding vector fields, as found above. There are several possible choices of three independent integrals in involution, which make the Kepler problem into a Liouville integrable system, for example $(H_1, \ell_i, |\ell|^2)$, $(H_1, \ell_i, A_i)$, or $(H_1, \ell_i, |A|^2)$.\textbf{}

\subsection{Example: Calogero-Moser system}

The Calogero-Moser (CM) system is governed by the Hamiltonian 
\[ H(q,p) = \frac{1}{2} |p|^2 + \frac{1}{2} \sum_{\beta=1}^n \sum_{\alpha=1}^{\beta-1} \frac{1}{(q^\alpha - q^\beta)^2} \,, \]
see for example \cite{calogero1971solution, moser1975three, olshanetsky1981classical}.
Introducing the notation
\[ F(q) = \left( \sum_{\alpha\neq1} (q^1 - q^\alpha)^{-3}, \ldots,  \sum_{\alpha\neq n} (q^n - q^\alpha)^{-3} \right) \]
we can write this as
\[ H(q,p) = \frac{1}{2} |p|^2 + \frac{1}{2} F(q) \cdot q \,. \]
Noting that $\der{}{q}( F(q) \cdot q) = -2F(q)$ we find the equations of motion $\ddot{q} = F(q)$

The CM system possesses a sequence of conserved quantities $I_j$, containing $H$ as $I_2$, which are pairwise in involution. In addition, there exist conserved quantities $K_j$, which make the system superintegrable \cite{wojciechowski1983superintegrability}. These conserved quantities can be constructed from the system's Lax pair, consisting of the $n \times n$ matrices $\boldsymbol{L}$ and $\boldsymbol{M}$ with entries
 \begin{align*}
 \boldsymbol{L} &= \begin{pmatrix}
 p_1 & \frac{i}{q_1-q_2} & \ldots & \frac{i}{q_1-q_n} \\
 \frac{i}{q_2-q_1} & p_2 & \ldots & \frac{i}{q_2-q_n} \\
 \vdots & \vdots & \ddots & \vdots \\
 \frac{i}{q_n-q_1} & \frac{i}{q_n-q_2} & \ldots & p_n \\
 \end{pmatrix} \,,
 \\
 \boldsymbol{M} &= \begin{pmatrix}
 - \sum_{k\neq1} \frac{i}{(q_1-q_k)^2} & \frac{i}{(q_1-q_2)^2} & \ldots & \frac{i}{(q_1-q_n)^2} \\
 \frac{i}{(q_2-q_1)^2} & - \sum_{k\neq2} \frac{i}{(q_2-q_k)^2} & \ldots & \frac{i}{(q_2-q_n)^2} \\
 \vdots & \vdots & \ddots & \vdots \\
 \frac{i}{(q_n-q_1)^2} & \frac{i}{(q_n-q_2)^2} & \ldots & - \sum_{k\neq n} \frac{i}{(q_1-q_k)^2} \\
 \end{pmatrix} \,,
 \end{align*}
and the additional matrix $\boldsymbol{N} = \mathrm{diag}(q^1,\ldots,q^n)$. They are
\[ I_j = \frac{1}{j} \tr(\boldsymbol{L}^j)\]
and
\[ K_j = \tr(\boldsymbol{N} \boldsymbol{L}^{j-1}) \tr(\boldsymbol{L}) - \tr(\boldsymbol{L}^j) \tr(\boldsymbol{N}) \,. \]

Denote  $e = (1, \ldots, 1)$. As a minimal example to illustrate our framework, we consider only three non-commuting Hamiltonians,
\begin{align*}
    & I_2(q,p) = H(q,p) \,, \\
    & I_1(q,p) = e \cdot p \,, \\
    & K_2(q,p) =  (q \cdot p) (e \cdot p) - \left( |p|^2 + F(q)\cdot q \right) (e \cdot q) \,.
\end{align*}
We have $\{I_2,K_1\} = 0$, $\{I_2,K_2\} = 0$, and $\{I_1,K_2\} = 2n H - I_1^2$, hence we are dealing with nonlinear Poisson relations.

Let $\h$ be the space of functions of $I_1, I_2, K_2$. As in Section \ref{sec-nonlin-poisson} we define $\cL$ by
\[ f \ip \cL = q_1 \cdot \partial_f q - f \]
for $f \in \h$. In particular, setting $\partial_1 = \partial_{H}$, $\partial_2 = \partial_{I_1}$, and $\partial_3 = \partial_{K_2}$ (such that $\partial_1$ corresponds to the flow of the Calogero-Moser equation itself), we have
\begin{align*}
    &L_1 := H \ip \cL = \frac{1}{2} |q_1|^2 - \frac{1}{2} F(q) \cdot q \,, \\
    &L_2 := I_1 \ip \cL = q_1 \cdot (q_2 - e) \,, \\
    &L_3 := K_2 \ip \cL = q_1 \cdot q_3 - (q \cdot q_1) (e \cdot q_1) + \left( |q_1|^2 + F(q)\cdot q \right) (e \cdot q) \,.
\end{align*}
Note that these are just three components of the 1-tensor $\cL$ on the infinite-dimensional space $\tilde \g$. The multi-time Euler-Lagrange equations yield
\begin{align*}
    &q_{11} = F(q) \,, \\
    &q_2 = e \,, \\
    &q_3 = (q \cdot q_1) e + (e\cdot q_1) q - 2 (e \cdot q) q_1 \,.
\end{align*}
The coefficients of $\d \cL$ corresponding to the generators $\partial_1$, $\partial_2$, and $\partial_3$ are
\begin{align*}
    \partial_1 L_2 - \partial_2 L_1
    &= (q_2 - e) \cdot (q_{11} - F(q)) \,, \\
    \partial_1 L_3 - \partial_3 L_1
    &= (q_{11} - F(q)) \cdot (q_3 + 2 (e \cdot q) q_1 - (q \cdot q_1) e - (e\cdot q_1) q ) \,,
\end{align*}
and
\begin{align*}
    &\partial_2 L_3 - \partial_3 L_2 + \{I_1,K_2\} \ip \cL \\
    &= q_{12} \cdot (q_3 + 2 (e \cdot q) q_1 - (q \cdot q_1) e - (e\cdot q_1) q ) \\
    &\quad - (q_2 - e) \cdot \left(q_{13} + 2 q_{11} (e \cdot q) + q_1 (e \cdot q_1) - e |q_1|^2 - e (q \cdot q_{11}) - q(e \cdot q_{11}) \right) \\
    &\quad + (q_2 - e) \cdot \left( 2 (e \cdot q) (q_{11} - F(q)) - (q \cdot (q_{11} - F(q)) ) e - (e \cdot (q_{11} - F(q)) ) q \right) \\
    &\quad + q_1 (q_{32} - q_{23} + \partial_{\{I_1,K_2\}} q ) \,,
\end{align*}
where the last term vanishes because it equals $-q_1 ( [\partial_{I_1},\partial_{K_2}] - \partial_{\{I_1,K_2\}} ) q = 0$. Hence we see that all three coefficients attain a double zero on solutions to the multi-time Euler-Lagrange equations.

\begin{remark}
Looking only at $H$, $I_1$ and $K_2$, one may argue that we can linearise the Poisson relations by rescaling the Hamiltonians, in which case we would not need the extension discussed in Section \ref{sec-nonlin-poisson} and we could give this small system a multiform structure on a finite-dimensional Lie group. Indeed, taking the following particular combinations
\begin{align*}
    & H(q,p) = \frac{1}{2} |p|^2 + \frac{1}{2} F(q) \cdot q \,, \\
    & I(q,p) = \frac{1}{2} I_1(q,p)^2 = \frac{1}{2} ( e \cdot p )^2 \,, \\
    & K(q,p) = \frac{K_2(q,p)}{I_1} = q \cdot p - \left( |p|^2 + F(q)\cdot q \right) \frac{e \cdot q}{e \cdot p} \,,
\end{align*}
the Poisson brackets between pairs of these functions are given by
$ \{H,I\} = \{H,K\} = 0$ and $\{I,K\} = 2 I - 2n H$ so that the corresponding vector fields form a Lie algebra. However, this seems rather specific to the small set of functions $H$, $I_1$ and $K_2$. We have not been able to establish such a linearisation procedure in general for a large collection of conserved quantities $I_j$, $K_j$ for the CM model.
\end{remark}

\section{Lagrangian 2-forms on Lie groups}
\label{sec-2form}

So far we have been dealing with Lagrangian 1-forms, which describe ODEs. Now we turn our attention to field theory, which in the simplest case is described by Lagrangian 2-forms.

Consider a first order Lagrangian 2-form $\cL[v]$ on a Lie group $G$ defined by
\[ \partial_{\xi_j} \ip \partial_{\xi_i} \ip \cL[v] = L_{ij}(v, v_1, \ldots, v_n) \,, \]
where $v_i = \partial_{\xi_i} v$ are the derivatives of the field $v: G \to Q$ along a basis vector $\xi_i$ of the Lie algebra of $G$.
In this case the multi-time Euler-Lagrange equations will be different from their counterparts for commuting flows. We adopt a similar approach as in Theorem \ref{thm-mtEL1} to obtain the following.

\begin{thm}\label{thm-mtEL2}
    Let $\cL$ be a 2-form depending on the first jet of $v$. The variational principle of Definition \ref{def-varpcple} is equivalent to
    \begin{subequations}
	\begin{align}
	& \der{L_{jk}}{v_j} + \der{L_{ki}}{v_i}  = 0 \,, \label{EL2-corner} \\
	& \der{L_{jk}}{v_\ell} = 0 \qquad \text{ for } \ell \neq j,k \label{EL2-alien} \,, \\
	& \var{jk}{L_{jk}}{v} + \sum_\ell \liec_{jk}^\ell p_\ell = 0 \,,  \label{EL2-plane}
	\end{align}
    \end{subequations}
	where $p_j = \der{L_{ij}}{v_i}$, which is well-defined (i.e. independent of $i$) in view of Equation \eqref{EL2-corner}, and
	\[ \var{jk}{}{v} = \der{}{v} - \D_j \der{}{v_j} - \D_k \der{}{v_k} \,. \] 
\end{thm}

\begin{proof}
	Fix $i < j < k$ and consider the function $P_{ijk}[v] = \partial_{\xi_k} \ip \partial_{\xi_j} \ip \partial_{\xi_i} \ip \d \cL[v]$ of the free jet bundle. Then 
	\begin{align*}
	P_{ijk} 
	&= \D_i (\partial_{\xi_k} \ip \partial_{\xi_j} \ip \cL) - \partial_{\xi_k} \ip \partial_{[\xi_i,\xi_j]} \ip \cL + \circlearrowleft_{ijk}\\
	&= \D_i L_{jk} + \sum_\ell \liec_{jk}^\ell L_{i \ell } + \circlearrowleft_{ijk} \,,
	\end{align*}
	where $\circlearrowleft_{ijk}$ denotes all terms obtained by cyclicly permuting $i,j,k$ in the previous terms.
	
	The vertical exterior derivative of $P_{ijk}$ (which gives the multi-time Euler-Lagrange equations using Lemma \ref{lemma-equivalence}) can be expanded into a sum over index-strings:
	\begin{equation*}\label{PinAj}\delta P_{ijk}= \sum_{\mathfrak{I}} A^\mathfrak{I} \delta v_\mathfrak{I}\end{equation*}
	with
	\begin{equation*}A^\mathfrak{I}=\der{P_{ijk}}{v_\mathfrak{I}}\,.\end{equation*}
	Using Lemma \ref{lemma-free2quotient} we can project to the quotiented jet bundle and write $\delta P_{ijk}$ as a sum over independent $\delta v_I$, where $I$ ranges over the set of multi-indices:
	\begin{equation*}\delta P_{ijk} = \sum_{I} B^I\delta v_I \,, \end{equation*}
	where the $B^I$ are given by Equation \eqref{A-tilde}. This allows us to find the multi-time Euler-Lagrange equations in terms of the $L_{ij}$ as follows.
	
	We have
	\begin{equation*}
	B^{ii}=\widetilde{ \der{P_{ijk}}{v_{ii}} } = \der{L_{jk}}{v_{i}} \,.
	\end{equation*}
	Setting this equal to zero, we obtain Equation \eqref{EL2-alien}. Similarly, the equations given by $B^{jj}$ and $B^{kk}$ give the same equation with the indices permuted.
	
	For $\ell \not\in \{i,j,k\}$, the coefficients $B^{i \ell}$ (in the case where $\ell>i$) and $B^{\ell i}$ (in the case where $\ell<i$) are given by 
	\begin{equation*}
	\widetilde{ \der{P_{ijk}}{v_{\ell i}} } = \der{L_{jk}}{v_{\ell}} \,.
	\end{equation*}
	Setting this to zero is again equivalent to Equation \eqref{EL2-alien}. 
	
	We next consider
	\begin{equation*}
	B^{ij} = \widetilde{ \der{P_{ijk}}{v_{ij}} } + \widetilde{ \der{P_{ijk}}{v_{ji}} } = \der{L_{ki}}{v_{i}}+\der{L_{jk}}{v_{j}} \,,
	\end{equation*}
	which leads to Equation \eqref{EL2-corner}, as do $B^{jk} = 0$ and $B^{ik} = 0$. Modulo equation \eqref{EL2-corner} we can define $p_j=\der{L_{ij}}{v_{i}}$. 
	
	Next, we find
	\begin{equation*}
	\begin{split}
	B^i &= \widetilde{ \der{P_{ijk}}{v_{i}} } + \sum_{\ell<m}\liec_{\ell m}^i \widetilde{ \der{P_{ijk}}{v_{m \ell}} } \\
	&= \der{L_{jk}}{v}+\D_i\der{L_{jk}}{v_{i}}+\D_j\der{L_{ki}}{v_{i}}+\D_k\der{L_{ij}}{v_{i}}
	+\sum_\ell \left(\liec_{ij}^\ell \der{L_{k \ell}}{v_{i}}+\liec_{jk}^\ell \der{L_{i \ell}}{v_{i}}+\liec_{ki}^\ell \der{L_{jk}}{v_{i}}  \right)\\
	&\qquad + \liec_{ij}^i \der{L_{ki}}{v_{j}}+\liec_{ik}^i \der{L_{jk}}{v_{k}}+\liec_{jk}^i \der{L_{ki}}{v_{k}} \,.
	\end{split}
	\end{equation*}
	By applying equations \eqref{EL2-alien} and \eqref{EL2-corner} to the above expression, we obtain that
	\begin{equation*}
	B^i=\var{jk}{L_{jk}}{v}+\sum_\ell \liec_{jk}^jp_\ell \,,
	\end{equation*}
	which we set equal to zero to obtain Equation \eqref{EL2-plane}. The same equation follows from $B^j = 0$ and $B^k = 0$.
	
	The equations  $B^\emptyset=0$ and $B^\ell=0$ for $\ell \notin \{i,j,k\}$ are consequences of Equations \eqref{EL2-alien}--\eqref{EL2-plane}.
\end{proof}

\begin{remark}[On the stepped surface approach] In the commutative case $G = \R^N$ the multi-time Euler-Lagrange equations can be obtained by a stepped surface approach \cite{suris2016lagrangian}. This consists in approximating any given surface by a stepped surface, i.e.\@ a piecewise flat surface where each flat piece is tangent to two coordinate directions. On a flat piece tangent to $\partial_{t_i}$ and $\partial_{t_j}$ the action integral only sees the coefficient $L_{ij}$ of the Lagrangian two-form. This leads to the Euler-Lagrange equation $\var{ij}{L_{ij}}{v_I} = 0$. The boundary terms that occur where different flat pieces meet, lead to the other multi-time Euler-Lagrange equations.

In the non-commutative case this approach breaks down. If $\xi_i$ and $\xi_j$ do not commute, then there may not exist any surface that is tangent to $\partial_{\xi_i}$ and $\partial_{\xi_j}$. Indeed, Frobenius' theorem indicates that such a surface only exists if $[\partial_{\xi_i},\partial_{\xi_j}]$ lies in the span of $\partial_{\xi_i}$ and $\partial_{\xi_j}$. Hence we cannot in general reproduce the essential property that individual pieces of a stepped surface only see one coefficient of the Lagrangian 2-form.

Irrespective of whether we are in the commutative or non-commutative case, the integral of a (Lagrangian) 2-form over a 2-dimensional surface is well-defined. Given a surface one can always endow it with local coordinates, say $x$ and $y$, and use the coordinate functions to pull back the 2-form to a subset of $\R^2$. This is the canonical way of defining the integral of a differential form. What we cannot do in general, is use the differential operators $\partial_x$ and $\partial_y$ as basis elements of the Lie algebra $\mathfrak{g}$, because $\partial_x$ and $\partial_y$ will not in general be left-invariant under $G$.
\end{remark}

\subsection{Example on $SE(2)$}
\label{sec-SE2}

In this subsection we consider an example of a Lagrangian 2-form with the Lie group $SE(2)$ as multi-time. We parameterise $SE(2)$ by $(x,y,\theta)$, such that its multiplication is given by Equation \eqref{SE2-product}.
Since we have global coordinates, we could take as multi-time the Euclidean space spanned by $\partial_x,\partial_y,\partial_\theta$, but in order to illustrate the non-commutative aspects of Lagrangian multiform theory, we consider a basis of left-invariant vector fields on $SE(2)$, given by
\begin{align*}
    &\partial_1 = \cos \theta \partial_x + \sin \theta \partial_y \,, \\
    &\partial_2 = - \sin \theta \partial_x + \cos \theta \partial_y \,, \\
    &\partial_3 = \partial_\theta \,.
\end{align*}
They satisfy
\[ [ \partial_i, \partial_j ] = \sum_{k=1}^3 \liec_{ij}^k \partial_k \]
with
\[ \liec_{32}^1 = \liec_{13}^2 = -1, \qquad
\liec_{23}^1 = \liec_{31}^2 = +1 , \]
and all other structure constants equal to zero.

We consider a Lagrangian 2-form $\cL[v]$ which describes harmonic functions and the rotational symmetry of the notion of harmonicity. It is defined by $\partial_j \ip \partial_i \ip \cL = L_{ij}$ with
\begin{align*}
& L_{12} = \frac{1}{2} v_1^2 + \frac{1}{2} v_2^2 \,, \\
& L_{13} = v_2 v_3 - y(\theta) v_1 v_2 - \frac{1}{2} x(\theta) (v_1^2 - v_2^2) \,, \\
& L_{23} = -v_1 v_3 - x(\theta) v_1 v_2 + \frac{1}{2} y(\theta) (v_1^2 - v_2^2) \,,
\end{align*}
where $x(\theta) = x \cos \theta + y \sin \theta$ and $y(\theta) = -x \sin \theta + y \cos \theta$. Note that $L_{13}$ and $L_{23}$ are non-autonomous and that we have
\begin{align*}
    &\partial_3 x(\theta) = y(\theta) \,, 
    &&\partial_1 x(\theta) = 1 \,,
    &&\partial_2 x(\theta) = 0 \,, \\
    &\partial_3 y(\theta) = -x(\theta) \,, 
    &&\partial_1 y(\theta) = 0 \,,
    &&\partial_2 y(\theta) = 1 \,.
\end{align*}

The multi-time Euler-Lagrange equations of type \eqref{EL2-corner} are trivially satisfied.
Those of type \eqref{EL2-alien} are
\begin{align*}
& 0 = \var{13}{L_{13}}{v_2} = v_3 - y(\theta) v_1 + x(\theta) v_2 \,, \\
& 0 = \var{23}{L_{23}}{v_1} = -v_3 - x(\theta) v_2 + y(\theta) v_1 \,,
\end{align*}
so we find the multi-time Euler-Lagrange equation
\begin{equation}\label{SE2-v3}
v_3 = - x(\theta) v_2 + y(\theta) v_1 \,.
\end{equation}
Differentiating this equation by $\partial_2$ and $\partial_1$, we have, respectively,
\begin{align*}
& v_{32} = y(\theta) v_{12} + v_1 - x(\theta) v_{22} \,, \\
& v_{31} = y(\theta) v_{11} - v_2 - x(\theta) v_{21} \,.
\end{align*}
Since $v_{ij} = \partial_j \partial_i v = \partial_i \partial_j v + [\partial_j,\partial_i] v = v_{ji} + \sum_k \liec_{ji}^k v_k$ it follows that
\begin{align}
&v_{23} = (y(\theta) v_{12} + v_1 - x(\theta) v_{22}) - v_1 = y(\theta) v_{12} - x(\theta) v_{22} \,, \label{SE2-v23} \\
&v_{13} = (y(\theta) v_{11} - v_2 - x(\theta) v_{21}) + v_2 = -x(\theta) v_{12} + y(\theta) v_{11} \,. \label{SE2-v13} 
\end{align}

There are three multi-time Euler-Lagrange equations of type \eqref{EL2-plane}. The first one is
\begin{equation}\label{SE2-v11}
0 = \var{12}{L_{12}}{v} = - v_{11} - v_{22} \,,
\end{equation}
because all the relevant structure constants are zero. The second one is
\begin{align*}
0 = \var{13}{L_{13}}{v} + \liec_{13}^2 \der{L_{12}}{v_1} 
&= \var{13}{L_{13}}{v} - \der{L_{12}}{v_1} \\
&= (-v_{23} + y(\theta) v_{12}  + x(\theta) v_{11} + v_1) - v_1 \\
&= (-v_{23} + y(\theta) v_{12} - x(\theta) v_{22}) + x(\theta) (v_{22} + v_{11}) \,,
\end{align*}
which is a consequence of equations \eqref{SE2-v23} and \eqref{SE2-v11}. The third one is
\begin{align*}
0 = \var{23}{L_{23}}{v} + \liec_{23}^1 \der{L_{21}}{v_2} 
&= \var{23}{L_{23}}{v} - \der{L_{21}}{v_2} \\
&= (v_{13} + x(\theta) v_{12} + y(\theta) v_{22} + v_2) - v_2 \\
&= (v_{13} + x(\theta) v_{12} - y(\theta) v_{11}) + y(\theta) (v_{11} + v_{22}) \,,
\end{align*}
which is a consequence of equations \eqref{SE2-v13} and \eqref{SE2-v11}.

The coefficient $P_{123}$ of exterior derivative of $\cL$ is
\begin{align*}
P_{123} 
&= \D_1 L_{23} - \sum_\ell \liec_{12}^\ell L_{\ell 3} + \circlearrowleft_{123} \\
&= \D_1 L_{23} + \circlearrowleft_{123} \\
&= (v_{11} + v_{22})(-v_3 + y(\theta) v_1 - x(\theta) v_2) + v_1 (v_{13} - v_{31}) + v_2 (v_{23} - v_{32}) \,.
\end{align*}
Due to the Lie algebra structure, the last two terms cancel against each other. The remaining term is a double zero on the multi-time Euler-Lagrange equations.

Solutions to the multi-time Euler-Lagrange equations \eqref{SE2-v3}--\eqref{SE2-v11} are harmonic functions with their rotations parameterised by $\theta$. For example, we could take
\[ v(x,y,\theta) = \exp(x(\theta)) \sin(y(\theta)) \,. \]
Its derivatives are
\begin{align*}
    &\partial_1 v =  \exp(x(\theta)) \sin(y(\theta)) = v \,, \\
    &\partial_2 v =  \exp(x(\theta)) \cos(y(\theta)) \,, \\
    &\partial_3 v = y(\theta) v_1 - x(\theta) v_2 \,,
\end{align*}
and similarly for higher derivatives. It is an easy calculation to see that $v$ satisfies the Euler-Lagrange equations.

\subsection{Infinite hierarchy example}

In this section, we construct an example of a multi-time with non-commuting vector fields by realising the idea initially proposed in \cite{nijhoff1986integrable}. We use the approach of \cite{caudrelier2021multiform} and \cite{caudrelier2022classical} to produce an example of a generalisation of the AKNS hierarchy \cite{ablowitz1974inverse} with non-commuting flows.

Let us first recall the main ingredients of the description in the commutative case, i.e.\@ the case where the multi-time is simply $\mathbb{R}^{\mathbb{N}}$. In this case, the Lagrangian $2$-form of interest is of the form
\[ \cL=\sum_{m<n}L^{mn}\, \d t_m\wedge \d t_n \]
and the Lagrangian coefficients $L^{mn}$ are most conveniently assembled into a \defn{generating Lagrangian multiform}, which is the following formal series in $\lambda^{-1}$, $\mu^{-1}$,
\begin{equation}
\label{L_series}
    L(\lambda,\mu)=\sum_{m,n=0}^\infty L^{mn}\,\lambda^{-m-1}\mu^{-n-1}\,,
\end{equation}
where $L^{mn} = -L^{nm}$.
In \cite{caudrelier2022classical} an expression for a large class of generating Lagrangian multiforms was introduced. In the present work, we focus on the following choice
\begin{equation}
\label{L_AKNS}
    L(\lambda,\mu)=\tr \left( \varphi(\mu)^{-1} \partial_\lambda \varphi(\mu) J - \varphi(\lambda)^{-1} \partial_\mu \varphi(\lambda) J \right)
- \frac{1}{\lambda - \mu} \tr \left( Q(\lambda) Q(\mu) - J^2 \right)
\end{equation}
and will explain how to extend it to a non-commutative setting. A few definitions are in order. The matrix $J$ is a constant element of the underlying Lie algebra $\g$.
 The object $\varphi(\lambda)$ is an element of the group with (matrix) Lie algebra $\g\otimes \lambda^{-1}\mathbb{C}[[\lambda^{-1}]]$, that is a matrix of the form
\begin{equation}
\label{form_varphi}
     \varphi(\lambda)=\1+\sum_{n\ge 1}\varphi^{(n)}\lambda^{-n}\,,
\end{equation}
 and 
\[ Q(\lambda)=\varphi(\lambda)\,J\,\varphi(\lambda)^{-1}\,. \]
The group element $\varphi(\lambda)$, or the algebra element $Q(\lda)$, contains the dynamical variables (fields) of the hierarchy and defines the phase space.
The generating derivation $\partial_\lambda$ is given by the formal series
\[ \partial_\lambda=\sum_{n=0}^\infty\partial_{t_n}\lambda^{-n-1}\,. \]
One can check that \eqref{L_AKNS} is of the form \eqref{L_series} and derive the Lagrangian coefficients $L^{mn}$ in terms of the dynamical variables.
If we choose $\g = \mathfrak{sl}(2)$ and take $J=-i\sigma_3$, \eqref{L_AKNS} produces all the coefficients for a Lagrangian multiform of the AKNS hierarchy \cite{caudrelier2021multiform,caudrelier2022classical}. Other choices of $\g$ would produce multicomponent generalisations of the AKNS hierarchy. 

The multi-time Euler-Lagrange equation associated to \eqref{L_AKNS} takes the form of a generating Lax equation,
\begin{align}
\label{gen_Lax}
\partial_\lambda Q(\mu)=  \left[\frac{1}{\lambda - \mu}Q(\lambda)  , Q(\mu) \right]\,.
\end{align}
This is the central equation obtained in \cite{flaschka1983kac}, written here in generating form (see also \cite{nijhoff1988integrable}). Of course, an important result is that the flows of the hierarchy commute, meaning that the generating vector fields satisfy $[\partial_\lambda,\partial_\mu]=0$ or, in other words that $[\partial_{t_n},\partial_{t_m}]=0$ for all $n,m\ge 0$. 

In this construction, there is a freedom in choosing the matrix $J$ defining the hierarchy. The effects of choosing another element of the underlying Lie algebra have been studied in detail in \cite{caudrelier2022classical}. With this is mind,  the idea promoted in \cite{nijhoff1986integrable} consists in attaching a hierarchy of vector fields to each possible choice of matrix $J$ and to consider the resulting hierarchies assembled in a single, enlarged hierarchy with not necessarily commuting vector fields. With the tools developed in the present paper, we are in a position to realise this idea in the context of Lagrangian multiforms as follows. 

Let us consider the generators $J_1, \ldots, J_N$ of a matrix Lie algebra $\mathfrak{g}$, satisfying
\begin{equation}
\label{J_algebra}
    [J_a, J_b] = \sum_c \liec_{ab}^c J_c \,.
\end{equation} 
As in Equation \eqref{form_varphi}, we denote by $\varphi$ an element of the group corresponding to the half loop algebra $\mathfrak{g}\otimes \lambda^{-1}\mathbb{C}[[\lambda^{-1}]]$ and define 
\[ Q_a(\lambda) = \varphi(\lambda) \,J_a\, \varphi(\lambda)^{-1}\,,~~a=1,\dots,N \,. \]
These elements $Q_a(\lda)$ contain the phase space variables of our hierarchy. 
To each $a \in \{1,\dots,N\}$, we associate a generating vector field 
\begin{equation}
\label{generating-vf}
\partial_{\lambda,a}=\sum_{n=0}^\infty\partial_{n,a}\lambda^{-n-1}\,,
\end{equation}
which acts on our phase space.
Note that we wrote $\partial_{n,a}$ on purpose and not $\partial_{t_{n,a}}$ since in general these vector fields do not commute and cannot be thought of as derivatives with respect to time coordinates $t_{n,a}$.
Instead, we will construct a hierarchy where these generating vector fields satisfy the Lie algebra relations
\begin{equation}
\label{vf_algebra}
[ \partial_{\lambda,a} , \partial_{\mu,b} ] = \sum_c \liec_{ab}^c \frac{\partial_{\lambda,c} - \partial_{\mu,c}}{\lambda - \mu}\,.
\end{equation} 
This should be understood as the generating form for the following algebra relations
\begin{equation}
\label{vf_algebra2}
[ \partial_{n,a} , \partial_{m,b} ] = \sum_c \liec_{ab}^c \,\partial_{n+m,c}\,,\qquad n,m\ge 0\,.
\end{equation} 

We can now define a generalisation of \eqref{L_AKNS} to the present context as follows.
Define the Lagrangian 2-form $\cL$ by $\partial_{\mu,b} \ip \partial_{\lambda,a} \ip \cL =L_{ab}(\lambda, \mu) $ where
\begin{equation}
\label{form_Lab}
    L_{ab}(\lambda, \mu) 
= \tr \left( \varphi(\mu)^{-1} \partial_{\lambda,a} \varphi(\mu) J_b - \varphi(\lambda)^{-1} \partial_{\mu,b} \varphi(\lambda) J_a \right)
- \frac{1}{\lambda - \mu} \tr \left( Q_a(\lambda) Q_b(\mu) - J_a J_b \right) .
\end{equation}  
The following lemma ensures that $    L_{ab}(\lambda, \mu) $ can be expanded as in \eqref{L_series}, 
\begin{equation}
\label{Lab_series}
    L_{ab}(\lambda,\mu)=\sum_{m,n=0}^\infty L_{ab}^{mn}\,\lambda^{-m-1}\mu^{-n-1}\,.
\end{equation}

\begin{lemma}\label{lemma-lm-divisible} The expression
$\tr \left( Q_a(\lambda) Q_b(\mu) - J_a J_b \right)$ in \eqref{form_Lab} is divisible by $\lambda^{-1} - \mu^{-1}$
\end{lemma}
\begin{proof}
We have
\begin{align*}
    \tr \left( Q_a(\lambda) Q_b(\mu) - J_a J_b \right)
    &= \tr \left( \varphi(\lambda) J_a \varphi(\lambda)^{-1} \varphi(\mu) J_b \varphi(\mu)^{-1} - J_a J_b \right) \\
    &= \tr \left( \left(\varphi(\mu)^{-1} \varphi(\lambda) - \1 \right) J_a \varphi(\lambda)^{-1} \varphi(\mu) J_b \right.\\
    &\qquad\qquad \left. + J_a ( \varphi(\lambda)^{-1} \varphi(\mu) - \1) J_b \right) \,.
\end{align*}
Now it remains to show that $(\varphi(\mu)^{-1} \varphi(\lambda) - \1)$ and $( \varphi(\lambda)^{-1} \varphi(\mu) - \1)$ are divisible by $\lambda^{-1} - \mu^{-1}$. We can expand $\varphi(\lambda)^{-1} \varphi(\mu)$ as
\begin{equation}\label{lm-inverse}
    \varphi(\lambda)^{-1} \varphi(\mu) = \sum_{n \geq 0} \sum_{j=0}^n F_{n,j} \lambda^{-j} \mu^{-(n-j)} \,,
\end{equation}
where $F_{0,0} = \1$. Setting $\lambda = \mu$ in equation \eqref{lm-inverse} we find that for all $n \geq 1$
\[ \sum_{j=0}^n F_{n,j} = 0 \,. \]
Hence
\begin{align*}
    \varphi(\lambda)^{-1} \varphi(\mu) &= \sum_{n \geq 0} \sum_{j=0}^n F_{n,j} (\lambda^{-j} \mu^{-(n-j)} - \mu^n) + \1 \\
    &= \sum_{n \geq 0} \sum_{j=0}^n F_{n,j} \mu^{-(n-j)} (\lambda^{-j}  - \mu^{-j}) + \1 \,.
\end{align*}
Similarly we also find that $(\varphi(\mu)^{-1} \varphi(\lambda) - \1)$ is  divisible by $\lambda^{-1} - \mu^{-1}$.
\end{proof}

The multi-time Euler Lagrange equations for $\cL$ so defined produce a non-commutative version of the generating Lax equation \eqref{gen_Lax}, controlled by the Lie algebra with Lie bracket given by \eqref{J_algebra}. Indeed, we have the following
\begin{prop}
The multi-time Euler Lagrange equations for $\cL$ take the form
\begin{equation}
\label{Q-evolution}
\begin{split}
    \partial_{\lambda,a} Q_b(\mu) &= \left[ \frac{Q_a(\lambda)-Q_a(\mu)}{\lambda - \mu} , Q_b(\mu) \right] \\
    &=  \frac{1}{\lambda - \mu} \left( \left[ Q_a(\lambda) , Q_b(\mu) \right] - \sum_{c} \liec_{ab}^c Q_c(\mu) \right) \,. 
\end{split}
\end{equation}
\end{prop}

\begin{proof}
Using the fact that partial derivatives of a trace are calculated as $\der{}{B^T} \tr(ABC) = CA$, we find
\[ 
\var{}{L_{ab}(\lambda, \mu) }{\varphi(\mu)^T}
= \left[ J_b , \varphi(\mu)^{-1} \partial_{\lambda,a} \varphi(\mu) \right] \varphi(\mu)^{-1} + \frac{1}{\lambda - \mu} \varphi(\mu)^{-1} \left[ Q_a(\lambda) , Q_b(\mu) \right]
\]
and
\[ 
p_b(\mu) := \der{L_{ab}(\lambda, \mu) }{(\partial_{\lambda,a} \varphi(\mu)^T)}
= J_b \varphi(\mu)^{-1} \,.
\]
The multi-time Euler-Lagrange equations of types \eqref{EL2-corner} and \eqref{EL2-alien} are trivially satisfied. Equation \eqref{EL2-plane} reads
\begin{equation}\label{hierarchyEL}
0 = \var{}{L_{ab}(\lambda,\mu)}{\varphi(\mu)} + \sum_{c}  \frac{\liec_{ab}^c}{\lambda - \mu} p_c(\mu) \,.
\end{equation}
Hence we find
\begin{align}
0 &=
\left[ J_b , \varphi(\mu)^{-1} \partial_{\lambda,a} \varphi(\mu) \right] \varphi(\mu)^{-1} + \frac{1}{\lambda - \mu} \varphi(\mu)^{-1} \left[ Q_a(\lambda) , Q_b(\mu) \right]
- \sum_{c} \frac{\liec_{ab}^c}{\lambda - \mu} J_c \varphi(\mu)^{-1} \notag\\
&= \left[ J_b , \varphi(\mu)^{-1} \partial_{\lambda,a} \varphi(\mu) \right] \varphi(\mu)^{-1} 
+ \frac{1}{\lambda - \mu} \varphi(\mu)^{-1} \left( \left[ Q_a(\lambda) , Q_b(\mu) \right] - \sum_{c} \liec_{ab}^c Q_c(\mu) \right) \notag\\
&= \left[ J_b , \varphi(\mu)^{-1} \partial_{\lambda,a} \varphi(\mu) \right] \varphi(\mu)^{-1} 
+ \frac{1}{\lambda - \mu} \varphi(\mu)^{-1} \left[ Q_a(\lambda) - Q_a(\mu), Q_b(\mu) \right] \,, \label{AKNS-EL}
\end{align}
where the second term can be expanded in negative powers of $\lambda$ and $\mu$ because $Q_a(\lambda) - Q_a(\mu)$ is divisible by $\lambda^{-1}-\mu^{-1}$.
As a consequence we find
\begin{align}
\partial_{\lambda,a} Q_b(\mu)
&= \partial_{\lambda,a} \varphi(\mu) J_b \varphi(\mu)^{-1} - \varphi(\mu) J_b \varphi(\mu)^{-1} \partial_{\lambda,a} \varphi(\mu) \varphi(\mu)^{-1} \notag\\
&= \varphi(\mu) \left[ \varphi(\mu)^{-1} \partial_{\lambda,a} \varphi(\mu), J_b \right] \varphi(\mu)^{-1} \notag\\
&=  \frac{1}{\lambda - \mu} \left[ Q_a(\lambda) - Q_a(\mu), Q_b(\mu) \right]
\end{align}
as claimed.
\end{proof}
Equation \eqref{Q-evolution} is a non-commutative generalisation of the generating Lax equation \eqref{gen_Lax} in the sense that the latter is obtained in the special case where we consider only one generator $J_a$ (which would be proportional to $\sigma_3$ in the historical example of \cite{flaschka1983kac}). In our more general context,  we have $N$ copies of \eqref{gen_Lax}, obtained when $b=a$, which are coupled with each other by the remaining equations when $b\neq a$.

A direct calculation shows that \eqref{Q-evolution} implies 
\begin{equation*}
[ \partial_{\lambda,a} , \partial_{\mu,b} ] Q_c(\nu)= \sum_d \liec_{ab}^d \frac{\partial_{\lambda,d} - \partial_{\mu,d}}{\lambda - \mu}\,Q_c(\nu) \,,
\end{equation*} 
so that we have a realisation of the Lie algebra \eqref{vf_algebra} as desired.
In fact, the calculations show the stronger result that the following deformation of the usual zero curvature equations for a hierarchy holds. Let us introduce
\begin{equation}
\label{def-Va}
    V_a(\lda,\nu)=\frac{Q_a(\lda)}{\lda-\nu}\,.
\end{equation}
The formal series expansion in this expression is understood as
\begin{equation}
\label{def_Lax_matrices}
    V_a(\lda,\nu)=\sum_{n=0}^\infty\frac{1}{\lda^{n+1}}V_a^{(n)}(\nu) \,, \qquad \text{where}~~V_a^{(n)}(\nu)=\sum_{k=0}^n\nu^kQ_a^{(n-k)}\,.
\end{equation}
This defines the Lax matrices $V_a^{(n)}(\nu)$ of the hierarchy, which are polynomials in $\nu$, and in terms of which we have: 
\begin{prop}\label{NZC}
The following non-commutative zero curvature equation, in generating form, holds:
\begin{equation}
\label{gen_ZC}
\partial_{\lambda,a}     V_b(\mu,\nu)  -\partial_{\mu,b}   V_a(\lda,\nu)+[V_b(\mu,\nu),V_a(\lda,\nu)]=\sum_c \liec_{ab}^c\frac{V_c(\lda,\nu)-V_c(\mu,\nu)}{\lda-\mu}\,.
\end{equation}
The coefficient of $\lda^{-n-1}\mu^{-m-1}$ gives the set of zero curvature equations for the non-commutative hierarchy, 
\begin{equation}
\label{set_ZC}
\partial_{n,a}  V_b^{(m)}(\nu)  -\partial_{m,b}   V_a^{(n)}(\nu)+[V_b^{(m)}(\nu),V_a^{(n)}(\nu)]= - \sum_c \liec_{ab}^cV_c^{(m+n)}(\nu)\,,    
\end{equation}
where $n,m\geq 0$.
\end{prop}
\begin{proof}
Equation \eqref{gen_ZC} is obtained by direct calculation from Equations \eqref{Q-evolution} and \eqref{def-Va}.

The coefficient of $\lda^{-n-1}\mu^{-m-1}$ in Equation \eqref{gen_ZC} gives Equation \eqref{set_ZC}. On the left hand side this is obvious, whereas on right hand side we use
\begin{align*}
    \frac{1}{\lda-\mu} \left( V_c(\lda,\nu)-V_c(\mu,\nu) \right)
    &= \frac{\lda^{-1} \mu^{-1}}{\mu^{-1} - \lda^{-1}} \sum_{n=0}^\infty \left( \lda^{-n-1} - \mu^{-n-1} \right) V_c^{(n)}(\nu) \\
    &= \sum_{n=0}^\infty \left( -\lda^{-n-1} \mu^{-1} -\lda^{-n} \mu^{-2} - \ldots - \lda^{-1} \mu^{-n-1} \right) V_c^{(n)}(\nu) \\
    &= - \sum_{m,n=0}^\infty \lda^{-n-1} \mu^{-m-1} V_c^{(m+n)}(\nu) \,. \qedhere
\end{align*}
\end{proof}

\begin{remark}
The terminology ``non-commutative zero curvature equations'' is used to suggest that our Equations \eqref{gen_ZC} and \eqref{set_ZC} generalise the usual zero curvature equations in the commutative case. The latter have the interpretation of encoding the flatness of a (Lax) connection or, equivalently, the commutativity of the components of a covariant derivative. On the other hand, Equations \eqref{gen_ZC} and \eqref{set_ZC} are not meant to be interpreted as such a flatness condition. However, they do possess a beautiful interpretation in terms of (generating) ``covariant vector fields'', defined as
\begin{equation*}
    {\cal D}_{\lda,a}(\nu)\equiv \partial_{\lda,a}-V_a(\lda,\nu)\,.
\end{equation*}
We do \emph{not} call these ``covariant derivatives'' because the vector fields $\partial_{n,a}$ should not be thought of as derivatives with respect to time variables $t_{n,a}$, as already advocated above.
Assuming the Lie algebra relations \eqref{vf_algebra} hold, \eqref{gen_ZC} is equivalent to
\begin{equation}
\label{cov_vf_algebra}
    [{\cal D}_{\lda,a}(\nu),{\cal D}_{\mu,b}(\nu)]= \sum_c \frac{\liec_{ab}^c}{\lambda - \mu}\left({\cal D}_{\lda,c}(\nu) - {\cal D}_{\mu,c}(\nu)\right) \,.
\end{equation}
This is the generalisation to our (generating) non-commutative context of the following well-known fact. If one has commuting vector fields, say $[\partial_x,\partial_t]=0$, and a (Lax) connection $U\,dx+V\,dt$, then the flatness or zero curvature condition is equivalent to the commutativity of the covariant derivatives,  $[\partial_x-U,\partial_t-V]=0$. In our setting, with Lie algebra relations \eqref{vf_algebra}, the ``non-commutative zero curvature equations'' \eqref{gen_ZC} are equivalent to the relations \eqref{cov_vf_algebra} on the ``covariant vector fields'' ${\cal D}_{\lda,a}(\nu)$ which realise once again the Lie algebra structure.
\end{remark}
\begin{remark}
There is a different way of using the generating zero curvature equation \eqref{gen_ZC} which was used in \cite{nijhoff1986integrable}. If we set $\lambda = \ell$, $\mu = \ell'$, $\nu=k$, $Q_a(\lda) = R_\ell$, $Q_b(\mu) = \widetilde{R}_{\ell'}$, and $[ R_\ell, \widetilde{R}_{\ell} ] = \widehat{R}_\ell$, and if we multiply \eqref{gen_ZC} by $(\ell-k)(\ell'-k)$ and require that it holds identically in $k$, then we obtain the pair of equations (43) from \cite{nijhoff1986integrable}:
\begin{subequations}
\begin{align}
&  \widetilde{\partial}_{\ell'} R_\ell - \partial_\ell \widetilde{R}_{\ell'}
=- \frac{\widehat{R}_{\ell} - \widehat{R}_{\ell'}}{\ell - \ell'} \,, \\
& \ell' \widetilde{\partial}_{\ell'} R_\ell - \ell \partial_\ell \widetilde{R}_{\ell'}
+ [R_\ell, \widetilde{R}_{\ell'}] + \frac{\ell' \widehat{R}_\ell - \ell \widehat{R}_{\ell'}}{\ell - \ell'} = 0 \,.
\end{align}
\label{R_eqs}
\end{subequations}%
In \cite{nijhoff1986integrable}, the commutative analogue of \eqref{R_eqs} led to a connection with the Wess-Zumino-Witten sigma model and it is an intriguing problem to cast \eqref{R_eqs} into a non-commutative version of this connection. 
\end{remark}

\begin{prop}[Closure relation]\label{closure_rel}
The form $\cL$ with coefficients \eqref{form_Lab} satisfies the closure relation. On the equations of motion \eqref{Q-evolution}, we have
\begin{equation}
    \partial_{\nu,c} \cL_{ab}(\lambda,\mu) +\partial_{\mu,b} \cL_{ca}(\nu,\lambda)+\partial_{\lambda,a} \cL_{bc}(\mu,\nu)  = 0 \,.
\end{equation}
\end{prop}
The proof is given in Appendix \ref{AppA}.

In the rest of this section we restrict our attention to $\g = \mathfrak{sl}(2)$ to illustrate how the usual (unreduced) nonlinear Schr\"odinger equation of the AKNS hierarchy sits within our non-commutative extension. It is enough to consider \eqref{set_ZC} for $n,m\in\{0,1,2\}$. To avoid confusing notations on the vector fields that would arise if we also used $a,b\in\{1,2,3\}$, we prefer to use the other common choice of $a,b\in\{x,y,z\}$ for the labels associated to the basis elements of $\mathfrak{sl}(2)$. Hence, instead of having vector fields such as $\partial_{1,2}$ and $\partial_{2,1}$ in Equation \eqref{generating-vf}, we will have $\partial_{1,y}$ and $\partial_{2,x}$. This shows more clearly what labels the level in the hierarchy and what labels the direction in the underlying Lie algebra. Let us emphasise again that $x,y,z$ should not be thought of as actual coordinates, since in general the vector fields $\partial_{n,x}$, $\partial_{m,y}$ and $\partial_{k,z}$ do not commute. 

As a basis of $\mathfrak{sl}(2)$ we choose
\[J_x=\begin{pmatrix}
0 & 1\\
1 & 0
\end{pmatrix}=\sigma_x\,,
\qquad
J_y=\begin{pmatrix}
0 & -i\\
i & 0
\end{pmatrix}=\sigma_y\,,
\qquad
J_z=\begin{pmatrix}
1 & 0\\
0 & -1
\end{pmatrix}=\sigma_z\,,\]
which gives us $\liec_{ab}^c=2i\epsilon_{abc}$ where $\epsilon_{abc}$ is the totally antisymmetric tensor with $\epsilon_{xyz}=1$.
We write the coefficients $\varphi^{(n)}$ of $\varphi(\lambda)$ in \eqref{form_varphi} as
\[ \varphi^{(n)}=\begin{pmatrix}
A_n & B_n\\
C_n & D_n
\end{pmatrix}\,. \]
Because we consider $\g=\mathfrak{sl}(2)$, we have $\det \varphi=1$ which implies that $D_1+A_1=0$ and $D_2+A_2+A_1D_1-B_1C_1=0$. We use this to eliminate $D_1$ and $D_2$.
This gives us 
\begin{align*}
& Q_x^{(0)}=\sigma_1\,,
\qquad
Q_x^{(1)}=\begin{pmatrix}
B_1-C_1 & 2A_1\\
-2A_1 & C_1-B_1
\end{pmatrix}\,,\\
& Q_x^{(2)}=\begin{pmatrix}
B_2-C_2-A_1(B_1+C_1) & A_1^2-B_1^2+2A_2\\
3A_1^2+2B_1C_1-C_1^2-2A_2 & C_2-B_2+A_1(B_1+C_1)
\end{pmatrix}\,,
\end{align*}
\begin{align*}
& Q_y^{(0)}=\sigma_2\,,
\qquad 
Q_y^{(1)}=\begin{pmatrix}
i(B_1+C_1) & -2iA_1\\
-2iA_1 & -i(B_1+C_1)
\end{pmatrix}\,,\\
& Q_y^{(2)}=\begin{pmatrix}
-i(A_1(B_1-C_1)-B_2-C_2) & -i(A_1^2+B_1^2+2A_2)\\
i(3A_1^2+2B_1C_1+C_1^2-2A_2)  & i(A_1(B_1-C_1)-B_2-C_2)
\end{pmatrix}\,,
\end{align*}
\begin{align*}
& Q_z^{(0)}=\sigma_3\,,
\qquad
Q_z^{(1)}=\begin{pmatrix}
0 & -2B_1\\
2C_1 & 0    
\end{pmatrix}\,, \\
& Q_z^{(2)}=\begin{pmatrix}
2B_1C_1 & -2A_1B_1-2B_2\\
-2A_1C_1+2C_2  & -2B_1C_1
\end{pmatrix}\,.
\end{align*}

We now spell out the content of \eqref{set_ZC} for the first three levels $n,m=0,1,2$ and write the equations they entail on the dynamical variables $A_1$, $B_1$, $C_1$. The level $n=0$, $m=0$ just gives the Lie algebra relations \eqref{J_algebra}.
At the level $n=0$, $m=1$, we have
\[ \partial_{0,a}  Q_b^{(1)}  +[Q_b^{(1)},\sigma_a]=-\sum_c \liec_{ab}^cQ_c^{(1)}\,. \]
Looking at the possible cases for $a,b$, this gives
\begin{subequations}
\begin{align}
\label{flow01}
&\partial_{0,x}A_1=C_1-B_1\,,
&&\partial_{0,x}B_1=-2A_1 \,,
&&\partial_{0,x}C_1=2A_1  \,,\\
\label{flow02}
&\partial_{0,y}A_1=-i(B_1+C_1)\,,
&&\partial_{0,y}B_1=2iA_1 \,,
&&\partial_{0,y}C_1=2iA_1\,,\\
\label{flow03}
&\partial_{0,z}A_1=0\,,
&&\partial_{0,z}B_1=2B_1 \,,
&&\partial_{0,z}C_1=-2C_1\,.
\end{align}
\end{subequations}%
In the standard AKNS case, at this first level we would only have \eqref{flow03} corresponding to $\sigma_z$, which is easily integrated. Here, to integrate the full set of equations would require the use of the whole group $SL(2)$.
For the higher flows, the first few of which are described next, this becomes even more complicated. One can imagine that a version of the usual dressing method could be implemented which will require the use of the group corresponding to the loop algebra $\mathfrak{sl}(2)\otimes \lambda^{-1}\mathbb{C}[[\lambda^{-1}]]$ but, contrary to the standard case, it is far from clear how explicitly the (seed) solutions can be constructed. This task is beyond the scope of the present paper. 

The next levels are studied in detail in Appendix \ref{derive}. They lead to 
\begin{subequations}\label{NC_NLS0}
\begin{eqnarray}
\label{NC_NLS1}
&&\partial_{1,b}  Q_a^{(1)}  +[Q_b^{(2)},\sigma_a]=0\,,\\     
\label{NC_NLS2}
 &&   \partial_{2,a}   Q_b^{(1)}-\partial_{1,a}  Q_b^{(2)}  +[Q_b^{(1)},Q_a^{(2)}]=0\,.
\end{eqnarray} 
\end{subequations}
Equations \eqref{NC_NLS0} contain the required information to cast the (unreduced) NLS equation into our non-commutative extension. 
As derived in Appendix \ref{derive}, with $B_1=-\frac12q$, $C_1=\frac12r$, $A_1=\frac12p$, we obtain
\begin{subequations}\label{x_flow}
\begin{align}
&\partial_{2,x}q-\frac{1}{2}\partial_{1,x}\partial_{1,z}q+q\partial_{1,x}p=0\,,\\
&\partial_{2,x}r-\frac{1}{2}\partial_{1,x}\partial_{1,z}r-r\partial_{1,x}p=0\,,\\
&\partial_{2,x}p-\frac{1}{2}\partial_{1,x}^2q+p\partial_{1,x}p-\frac{1}{2}(q+r)\partial_{1,x}q=0\,,
\end{align}
\end{subequations}
as well as similar equations for $\partial_{2,y}$
\begin{subequations}\label{y_flow}
\begin{align}
&\partial_{2,y}q-\frac{1}{2}\partial_{1,y}\partial_{1,z}q+q\partial_{1,y}p=0\,,\\
&\partial_{2,y}r+\frac{1}{2}\partial_{1,y}\partial_{1,z}r-r\partial_{1,y}p=0\,,\\
&\partial_{2,y}p-\frac{i}{2}\partial_{1,y}^2q+p\partial_{1,y}p+\frac{1}{2}(q-r)\partial_{1,y}q=0\,,
\end{align}
\end{subequations}
and for $\partial_{2,z}$
\begin{subequations}\label{z_flow}
\begin{align}
\label{NLS1}
&\partial_{2,z}q-\frac{1}{2}\partial_{1,z}^2q+q^2r=0\,,\\
\label{NLS2}
&\partial_{2,z}r+\frac{1}{2}\partial_{1,z}^2q-qr^2=0\,,\\
&\partial_{2,z}p-\frac{1}{2}r\partial_{1,z}q-\frac{1}{4}\partial_{1,z}(\partial_{1,x}+i\partial_{1,y})q+pqr=0\,.
\end{align}
\end{subequations}

When taken on their own, Equations \eqref{NLS1}-\eqref{NLS2} produce the unreduced NLS system. This is seen by setting
$2\partial_{2,z}=i\partial_T$, $\partial_{1,z}=i\partial_X$ to obtain 
\[i\partial_T q + \partial_X^2 q+2qrq=0\,,
\qquad
-i\partial_T r+  \partial_X^2 r-2rqr=0\,.\]
The familiar nonlinear Schr\"odinger equation,
\[iq_T + q_{XX}\pm2|q|^2q=0 , \] 
is obtained by applying the reduction $r=\pm q^*$. Equations \eqref{NLS1}--\eqref{NLS2} are the lowest nonlinear ones in the usual AKNS hierarchy of vector fields $\partial_{n,z}$ associated to $J_z=\sigma_z$. In our construction, it is now part of the larger set of equations \eqref{x_flow}-\eqref{z_flow} which play the similar role of being the lowest nonlinear equations in our non-commutative AKNS hierarchy of vector fields $\partial_{n,a}$. All these equations, including their higher counterparts contained in \eqref{set_ZC}, derive from the multi-time Euler Lagrange equations for our multiform $\cL$. It is important to note that, unlike the commutative case where one can consider \eqref{NLS1}-\eqref{NLS2} without any reference to the entire hierarchy,  it is not clear at this stage whether the truncated system \eqref{x_flow}-\eqref{z_flow} can be studied as a set of equations without the rest of the equations \eqref{set_ZC} in the hierarchy. This is because of the (infinite dimensional) Lie algebra relations \eqref{vf_algebra2} on the vector fields.
If $a\neq b$, this may mean that the ``differential consequences'' of \eqref{x_flow}-\eqref{z_flow} call upon higher equations not present in the truncation. This is an intriguing new feature which is not present in the usual case where $a=b=z$ and the flows of different levels commute. Its full understanding requires to investigate the solutions of the non-commutative AKNS hierarchy. As mentioned above, this is left for future work.

\section{Conclusions}

The main purpose of this work has been to extend the ideas of Lagrangian multiforms and pluri-Lagrangian systems to apply to non-commuting flows. The resulting theory of Lagrangian multiforms on Lie groups allows us to capture the full symmetry group of a system in the variational description, no matter if this group is abelian or not. To our knowledge, this is the first attempt at a Lagrangian theory of Lie group actions on manifolds.

One application of the theory of Lagrangian multiforms on Lie groups is to provide a variational description of superintegrable systems. On the other hand, the development of this theory shows that the notion of Lagrangian multiforms should not be constrained to the context of integrable systems. Indeed, it can be applied to any Lagrangian system with symmetries and captures both the dynamical system and its symmetries in a single variational principle.

Some questions remain on the topic of non-commuting hierarchies of PDEs. These include the proper interpretation of a differential form in generating form $L(\lambda,\mu)$ and the derivation of exact solutions to the non-commutative generalisation of the AKNS hierarchy, possibly using a version of the dressing method. An additional topic for future research, given that previous results have shown a remarkable similarity between the continuous and discrete theories of Lagrangian multiforms, is to develop a theory discrete Lagrangian multiforms for non-commuting maps and partial difference equations.

\paragraph{Acknowledgement}

The majority of this work was carried out while MV was a research fellow at the University of Leeds, supported by a Research Fellowship of the Deutsche Forschungsgemeinschaft (project number VE 1211/1-1). FN is currently supported by the EPSRC grant EP/W007290/1. 

\paragraph{Declaration} 

The authors have no competing interests to declare that are relevant to the content of this article.

\section*{Appendix}

\appendix

\section{Proof of Proposition \ref{closure_rel}}\label{AppA}
The coefficients of $\d \cL$ are generated by
\[ \partial_{\nu,c} \cL_{ab}(\lambda,\mu) + \circlearrowleft \,, \]
where $\circlearrowleft$ denotes the two terms obtained by cyclic permutations of $(a,b,c)$ and $(\lambda,\mu,\nu)$.
We write $L_{ab}(\lambda,\mu) = K_{ab}(\lambda,\mu) - V_{ab}(\lambda,\mu)$ with
\begin{align*}
    &K_{ab}(\lambda,\mu) = \tr \left( \varphi(\mu)^{-1} \partial_{\lambda,a} \varphi(\mu) J_b - \varphi(\lambda)^{-1} \partial_{\mu,b} \varphi(\lambda) J_a \right) \,, \\
    &V_{ab}(\lambda,\mu) = \frac{1}{\lambda - \mu} \tr \left( Q_a(\lambda) Q_b(\mu) - J_a J_b \right) \,,
\end{align*}
and analyse the contributions in turn. Using the Euler-Lagrange equation \eqref{AKNS-EL} and the cyclic property of the trace repeatedly, we find
\begin{align*}
\partial_{\nu,c} K_{ab}(\lambda,\mu) + \circlearrowleft
&= \frac{1}{\lambda - \mu} \frac{1}{\nu - \mu}\tr\left( \left[ Q_c(\nu) - Q_c(\mu), Q_b(\mu) \right] \left(  Q_a(\lambda) - Q_a(\mu) \right) \right) + \circlearrowleft \,.
\end{align*}
Note that $\tr\left( \left[ Q_c(\nu), Q_b(\mu) \right] Q_a(\lambda) \right)$ and $\tr\left( \left[ Q_c(\mu), Q_b(\mu) \right] Q_a(\mu) \right) = \tr\left( \left[ J_c, J_b \right] J_a \right)$ are invariant under cyclic permutations and that $\frac{1}{\lambda - \mu} \frac{1}{\nu - \mu} + \circlearrowleft = 0$. Hence
\begin{align}
\partial_{\nu,c} K_{ab}(\lambda,\mu) + \circlearrowleft
&= \frac{1}{\lambda - \mu} \frac{1}{\nu - \mu}\tr\left( -\left[ Q_c(\mu), Q_b(\mu) \right] Q_a(\lambda) - \left[ Q_c(\nu) , Q_b(\mu) \right] Q_a(\mu) \right) + \circlearrowleft \notag \\
&= \frac{1}{\lambda - \mu} \frac{1}{\nu - \mu}\tr\left( Q_a(\lambda) \left[ Q_b(\mu), Q_c(\mu) \right]  + \left[  Q_a(\mu) , Q_b(\mu) \right] Q_c(\nu) \right) + \circlearrowleft \,. \label{K-cyclic}
\end{align}
Similarly, we have
\begin{align*}
\partial_{\nu,c} V_{ab}(\lambda,\mu) + \circlearrowleft
&= \frac{1}{\lambda - \mu} \frac{1}{\nu - \lambda} \tr\left( \left[ Q_c(\nu) - Q_c(\lambda), Q_a(\lambda) \right] Q_b(\mu) \right) \\
&\qquad + \frac{1}{\lambda - \mu} \frac{1}{\nu - \mu} \tr\left( \left[ Q_c(\nu) - Q_c(\mu), Q_b(\mu) \right] Q_a(\lambda)  \right) + \circlearrowleft \,.
\end{align*}
Using again that $\tr\left( \left[ Q_c(\nu), Q_b(\mu) \right] Q_a(\lambda) \right)$ and $\tr\left( \left[ Q_c(\mu), Q_b(\mu) \right] Q_a(\mu) \right)$ are invariant under cyclic permutations, and that $\frac{1}{\lambda - \mu} \frac{1}{\nu - \mu} + \circlearrowleft = 0$, we find
\begin{align}
\partial_{\nu,c} V_{ab}(\lambda,\mu) + \circlearrowleft
&= \frac{1}{\lambda - \mu} \frac{1}{\nu - \mu} \tr\left(\left[ Q_a(\mu), Q_b(\mu) \right] Q_c(\nu) + Q_a(\lambda) \left[ Q_b(\mu), Q_c(\mu) \right]  \right) + \circlearrowleft \,.
\label{V-cyclic}
\end{align}
Comparing Equations \eqref{K-cyclic} and \eqref{V-cyclic} we obtain $\partial_{\nu,c} \cL_{ab}(\lambda,\mu) + \circlearrowleft\, = 0$ as desired.

\section{Derivation of Equations \eqref{x_flow}-\eqref{z_flow}}\label{derive}

At the level $n=1$, $m=1$, Equations \eqref{set_ZC} give one automatically satisfied relation
\[ [\sigma_b,Q_a^{(1)}]+[Q_b^{(1)},\sigma_a]=-\liec_{ab}^cQ_c^{(1)} \]
and the equation
\begin{equation}
\label{nm_11}
\partial_{1,a}  Q_b^{(1)}  -\partial_{1,b}   Q_a^{(1)}+[Q_b^{(1)},Q_a^{(1)}]=-\liec_{ab}^cQ_c^{(2)}\,,
\end{equation}
where summation over repeated indices is implied.
At the level $n=1$, $m=2$ we have
\begin{eqnarray}
\label{nm_12a}&&\partial_{1,a}  Q_b^{(1)}  +[Q_b^{(1)},Q_a^{(1)}]+[Q_b^{(2)},\sigma_a]=-\liec_{ab}^cQ_c^{(2)}\,, \\
\label{nm_12b}&&  \partial_{1,a}  Q_b^{(2)}  -\partial_{2,b}   Q_a^{(1)}+[Q_b^{(2)},Q_a^{(1)}]=-\liec_{ab}^cQ_c^{(3)}\,. 
\end{eqnarray}
We can combine \eqref{nm_11} and \eqref{nm_12a} to obtain
\begin{equation}
\label{nm_11b}
\partial_{1,b}  Q_a^{(1)}  +[Q_b^{(2)},\sigma_a]=0\,.     
\end{equation}
In a similar spirit, we can combine \eqref{nm_12b} with the equation coming from the coefficient of $\nu$ in Equation \eqref{set_ZC} with $n=2$, $m=2$, which reads
\[ \partial_{2,a}  Q_b^{(1)}  -\partial_{2,b}   Q_a^{(1)}+[Q_b^{(1)},Q_a^{(2)}]+[Q_b^{(2)},Q_a^{(1)}]=-\liec_{ab}^cQ_c^{(3)}\,. \]
This allows us to eliminate $Q_c^{(3)}$ and to obtain
\begin{equation}
\label{NC_NLS}
    \partial_{2,a}   Q_b^{(1)}-\partial_{1,a}  Q_b^{(2)}  +[Q_b^{(1)},Q_a^{(2)}]=0\,.
\end{equation} 
There are various equivalent ways of writing the equations on the fields $A_1$, $B_1$ and $C_1$ contained in \eqref{nm_11b} and \eqref{NC_NLS}. It is instructive to recall how one would proceed in the case $a=b=z$ to obtain the unreduced NLS system \eqref{NLS1}-\eqref{NLS2}. We have
\begin{eqnarray*}
&&\partial_{1,z} Q_z^{(1)}  + [Q_z^{(2)},\sigma_z]=0\,,\\
&&\partial_{2,z}   Q_z^{(1)}- \partial_{1,z}  Q_z^{(2)}  +[Q_z^{(1)},Q_z^{(2)}]=0\,.
\end{eqnarray*}
The first equation tells us that $-2A_1B_1-2B_2=-\partial_{1,z} B_1$ and $-2A_1C_1+2C_2=-\partial_{1,z} C_1$. 
Inserting into the second equation, the diagonal elements are automatically satisfied while the off-diagonal elements give 
\begin{subequations}
\label{NLS_unreduced}
\begin{align}
    -2\partial_{2,z}B_1 + \partial_{1,z}^2 B_1+8B_1^2C_1 &= 0\,, \\
    2\partial_{2,z}C_1 +  \partial_{1,z}^2 C_1+8B_1C_1^2 &= 0\,.
\end{align}
\end{subequations}
With $q=-2B_1$, $r=2C_1$, $2\partial_{2,z}=i\partial_T$, $\partial_{1,z}=i\partial_T$, this is \eqref{NLS1}-\eqref{NLS2}.

We apply a similar strategy for the rest of the equations obtained for $(a,b) \neq (z,z)$. We use the off-diagonal elements of \eqref{nm_11b} with $a=b=x$ to obtain the relation
\begin{equation}
\label{B2-C2}
    C_2-B_2+A_1(B_1+C_1)=\partial_{1,x}A_1\,.
\end{equation}
Then, the off-diagonal elements of \eqref{NC_NLS} with $a=x$, $b=z$ give the action of the vector field $\partial_{2,x}$ on $B_1$, $C_1$,
\begin{subequations}
\begin{align}
    2\partial_{2,x}B_1-\partial_{1,x}\partial_{1,z}B_1+4B_1\partial_{1,x}A_1 &=0\,, \\
    2\partial_{2,x}C_1+\partial_{1,x}\partial_{1,z}C_1-4C_1\partial_{1,x}A_1 &=0\,.
\end{align}
\end{subequations}
The diagonal elements of \eqref{NC_NLS} with $a=x$, $b=z$ are automatically satisfied upon taking into account \eqref{nm_11b} with $b=x$ and $a=z$.

We use the off-diagonal elements of \eqref{nm_11b} with $a=b=y$ to obtain the relation
\[C_2+B_2+A_1(C_1-B_1)=i\partial_{1,y}A_1\,.\]
Then, the off-diagonal elements of \eqref{NC_NLS} with $a=y$, $b=z$ give the action of the vector field $\partial_{2,y}$ on $B_1$, $C_1$,
\begin{subequations}
\begin{align}
    -2\partial_{2,y}B_1+\partial_{1,y}\partial_{1,z}B_1-4B_1\partial_{1,y}A_1 &=0\,, \\
    2\partial_{2,y}C_1+\partial_{1,y}\partial_{1,z}C_1-4C_1\partial_{1,y}A_1 &=0\,.
\end{align}
\end{subequations}
The diagonal elements of \eqref{NC_NLS} with $a=y$, $b=z$ are automatically satisfied upon taking into account \eqref{nm_11b} with $b=y$ and $a=z$.
We have already obtained the action of the vector field $\partial_{2,z}$ on $B_1$, $C_1$ in Equation \eqref{NLS_unreduced}, which was the (unreduced) NLS equation. 

It remains to obtain the actions of the vector fields $\partial_{2,a}$ on $A_1$. This can be done in several equivalent ways. The consistency is ensured by the equations relating the vector fields $\partial_{n,a}$ for different values of $a$. 
Let us first look at the off-diagonal elements in
\begin{equation}
\partial_{2,x}   Q_x^{(1)}-\partial_{1,x}  Q_x^{(2)}  +[Q_x^{(1)},Q_x^{(2)}]=0\,.
\end{equation}
We use \eqref{B2-C2} and \eqref{nm_11b} with $b=x$ and $a=z$, which gives $2A_2=B_1^2-A_1^2-\partial_{1,x}B_1$, to eliminate $B_2-C_2$ and $A_2$. We obtain
\begin{equation}
2\partial_{2,x}A_1+\partial_{1,x}^2B_1+4A_1\partial_{1,x}A_1-2(B_1-C_1)\partial_{1,x}B_1=0\,.
\end{equation}
Now, let us look at the off-diagonal elements in
\begin{equation}
        \partial_{2,y}   Q_y^{(1)}-\partial_{1,y}  Q_y^{(2)}  +[Q_y^{(1)},Q_y^{(2)}]=0\,,
\end{equation}
and use again \eqref{nm_11b} appropriately  to eliminate $A_2$, $B_2$, $C_2$. We obtain
\begin{equation}
     2\partial_{2,y} A_1+i\partial_{1,y}^2B_1+4A_1\partial_{1,y}A_1+2(B_1+C_1)\partial_{1,y}B_1=0\,.
\end{equation}
Finally, let us look at the two equations 
\begin{subequations}
\begin{align}
        & \partial_{2,z}   Q_x^{(1)}-\partial_{1,z}  Q_x^{(2)}  +[Q_x^{(1)},Q_z^{(2)}]=0\,, \\
        & \partial_{2,z}   Q_y^{(1)}-\partial_{1,z}  Q_y^{(2)}  +[Q_y^{(1)},Q_z^{(2)}]=0\,.
\end{align}
\end{subequations}
Multiplying the first one by $i$ and subtracting the second one, the $(1,2)$-entry gives
\begin{eqnarray}
4i\partial_{2,z}   A_1-2i\partial_{1,z} (A_1^2+2A_2)-8iA_1B_1C_1+8iB_2C_1=0\,.
\end{eqnarray}
As before, we can use \eqref{nm_11b} for appropriate choices of $a,b$ to obtain $2B_2=\partial_{1,z}B_1-2A_1B_1$ and $4A_2=-2A_1^2-(\partial_{1,x}+i\partial_{1,y})B_1$ and eliminate $B_2$ and $A_2$. This yields
\begin{eqnarray}
\partial_{2,z}   A_1+C_1\partial_{1,z}B_1+\frac{1}{4}\partial_{1,z}(\partial_{1,x}+i\partial_{1,y})B_1-4A_1B_1C_1=0\,.
\end{eqnarray}
Thus we have established that the vector fields $\partial_{2,a}$, acting on $A_1$, $B_1$ and $C_1$, give \eqref{x_flow}-\eqref{z_flow}.

\bibliographystyle{abbrvnat_mv}
\bibliography{comm}

\end{document}